\documentclass[pdflatex,sn-mathphys-num]{sn-jnl}

\usepackage{float}
\usepackage{graphicx}%
\usepackage{physics}
\usepackage{multirow}%
\usepackage{multicol}%
\usepackage{amsmath,amssymb,amsfonts}%
\usepackage{amsthm}%
\usepackage[title]{appendix}%
\usepackage{xcolor}%
\usepackage{textcomp}%
\usepackage{manyfoot}%
\usepackage{booktabs}%
\usepackage{algorithm}%
\usepackage{algorithmicx}%
\usepackage{algpseudocode}%
\usepackage{listings}%
\usepackage[utf8]{inputenc}
\usepackage[T1]{fontenc}
\usepackage{newtxtext,newtxmath}
\usepackage[scr=rsfso]{mathalfa}
\DeclareMathAlphabet{\mathscr}{U}{rsfso}{m}{n}%
\usepackage{tabularray}
\usepackage{caption}
\usepackage{academicons}
\usepackage{hyperref}
\usepackage{textpos}
\usepackage{tabularray}
\usepackage[justification=centering, singlelinecheck=false]{caption}
\usepackage{comment}

\theoremstyle{thmstyleone}%
\newtheorem{theorem}{Theorem}
\theoremstyle{thmstyletwo}%

\theoremstyle{thmstylethree}%

\begin{document}

\title[Microenvironmental Dynamics of Adenoid Cystic Carcinoma in Lacrimal Gland]
{Modeling the Microenvironmental Dynamics of Adenoid Cystic Carcinoma in Lacrimal Gland with Treatment Intervention }


\author[1]{\fnm{Ayasa} \sur{Akter}}
\author[1]{\fnm{Miss Sajnin Binta Salam} \sur{Efty}}
\author*[1,]{\fnm{Uzzwal Kumar} \sur{Mallick}}\email{mallickuzzwal@math.ku.ac.bd}

\affil[1]{\orgdiv{Mathematics Discipline}, \orgname{Khulna University}, \orgaddress{\street{Khulna-9208}, \country{Bangladesh}}}

\abstract{Adenoid cystic carcinoma (ACC) of the lacrimal gland is a rare and extremely life-threatening malignant tumor with a high risk of recurrence and poor long-term outcomes. On the other hand, due to its complex molecular dynamics and low incidence rate, treatment strategies are very limited. In this study, a mathematical model is developed based on the biological phenomena of the tumor–microenvironment using nonlinear ordinary differential equations. It is analyzed both analytically and numerically. According to the findings, the growth of tumor cells is associated with the depletion of oxygen and nutrients, while severe hypoxia promotes the development of abnormal blood vessels and makes tumor control more difficult. Surgery and postoperative radiotherapy have been considered as treatment strategies, and we have observed that surgery alone may not be sufficient to prevent the regrowth of residual tumor cells. In addition, postoperative radiotherapy provides better control of the residual tumor burden, particularly when the remaining tumor cells have higher radiosensitivity. Furthermore, different radiotherapy interruption scenarios have been investigated, and the results indicate that longer gaps between radiotherapy doses can adversely affect the overall treatment response, leading to increased overall toxicity and a higher number of residual tumor cells at the end of the observation period. Consequently, our results suggest that postoperative radiotherapy after surgery, combined with adequate radiosensitivity of the residual tumor cells and adherence to the prescribed treatment schedule without prolonged interruptions, may reduce the risk of tumor recurrence and potentially improve long-term survival.}



\maketitle
\section{Introduction}
Adenoid Cystic Carcinoma is one of the most frequent malignant tumors of the lacrimal gland. The tumor is relatively rare and characterized by its slow-growing but aggressive nature. Though the proliferation rate is slow, it infiltrates into peripheral tissues like orbital structures and the brain. The most frustrating nature of the tumor is metastasis. Mostly, ACC metastasizes into distant organs. For this reason, the long-term survival rate is very unsatisfactory. In terms of clinical characteristics like hypoxia, hypovacularity, perineural invasion and distant metastasis made it life-threatening. Although the tumor can be controlled locally, distant metastasis into the lungs, bones and liver after a certain time is one of the terrifying features. In this case, surgical management along with radiation therapy works well in controlling the tumor locally, but they have no effect on survival rates and metastasis. Targeted therapies such as chemotherapy have little more effectiveness on the tumor.  But due to a hypoxic microenvironment, chemotherapy also loses its efficacy sometimes. Besides, the irregular vasculature structure inhibits drug delivery to the cells. According to \cite{incidence2013}, the incidence rate of ACC was $1.3$ cases per million people from $1974-2007$ in Denmark. And later Powell et al. \cite{incidence2023}  investigated the incidence rate, which is $3-4$ cases per million people worldwide and these studies substantiate the rarity of the tumor. They also found that $5$ years survival rate is less than $50\%$ and in the case of 10 years, it decreases to $20\%$ in the studies \cite{survival1982} and \cite{survival1998}. In addition, Zeidan et al. \cite{liver_metastasis} reported a rare case of metastasis into the liver from Lacrimal gland ACC after $20$ years of removal of the primary tumor. Despite this rarity, the notable mortality rate and recurrence patterns underscore the clinical significance of ACC and the need for further research.\\ 

\noindent Since there are so many organs associated with our eyes that help in their functioning. The lacrimal gland is one of them. It is a compound tubulo-acinar exocrine gland with the key responsibility of secreting aqueous components, electrolytes and protective proteins of the tear film, which protects the ocular surface. It is located in the superolateral part of the eye orbit, above the eyeball, specifically tucked into the lacrimal fossa of the frontal bone. Conceptually, a tumor is a mass of abnormal tissue that proliferates in an uncontrolled manner. In general, tumors are classified into two categories: Benign and Malignant tumors. Benign tumors usually refer to non-cancerous and slow-growing tumors without spreading. In contrast, Malignant tumors are slow-growing, but they continue to grow relentlessly. Mostly, it is cancerous, invades into nearby tissues, spreads to distant organs, and comes back even after treatment. For instance, Pleomorphic adenoma is the most common benign tumor of the lacrimal gland and Adenoid Cystic Carcinoma and Carcinoma ex Pleomorphic Adenoma are malignant tumors.\\

\noindent 
This malignancy progresses gradually and poses various types of symptoms as it spreads. As the tumor gradually enlarges, it exerts pressure on surrounding tissues, leading to persistent orbital pain. In some cases, rapid tumor expansion may cause acute swelling of the upper eyelid and orbital region. With the advancement of the disease, infiltration of extraocular muscles and surrounding connective tissues may result in restricted ocular motility, leading to partial or complete loss of eye movement and diplopia. Also, the continued growth of the tumor mass displaces the eyeball from its normal anatomical position. Consequently, many patients develop proptosis, characterized by a forward displacement of the globe, which is one of the most common and clinically noticeable signs of lacrimal gland malignancy. In advanced stages, involvement of adjacent orbital structures may further exacerbate visual impairment, cause mechanical distortion of the globe and compromise normal ocular function. These progressive symptoms reflect the aggressive and infiltrative nature of lacrimal gland ACC.\\

\noindent 
Over the years, myriad studies have examined how tumor cells grow, their growth pattern, and the main cause and biological behavior of ACC. A modern genetic analysis by Dillon et al. \citep{Dillon2015} investigated that the main cause of ACC is the translocation of chromosome $6:9$, which induces $MYB-NFIB$ gene fusion, leading to the up-regulation of $MYB$ oncoprotein. Also, ACC can be differentiated from other malignancies of Lacrimal gland that is, Pleomorphic Adenoma (PA), by the presence of Proteomes in tears, which helps to detect ACC accurately \citep{Yue2023}. Argyris et al. \cite{PP2013} described the biological behavior of ACC as aggressive in nature and does not transmit highly. Another study by Lin et al \citep{Lin2015} identified ACC as a slow-growing tumor with widespread infiltration into surrounding tissues and with poor clinical outcome, which also penetrates into intracranial structures. Furthermore, a detached cell line named (LACC-1) has been established by them to analyze the growth pattern mathematically. They separated the tumor from the human lacrimal gland and left it growing in a culture and observed the growth pattern of the tumor in terms of a growth curve. But in reality, cells never grow in a curve due to a lack of resources; there are many factors associated with tumor growth that are absent in culture.\\ 

\noindent 
Another important aspect associated with the tumor microenvironment is the Nutrient level. To support growth, tumor cells must extort glucose, lipids and amino acids at a hastening rate from extracellular nutrients, investigated by Finicle et al \cite{Finicle2018}. Ghaffari et al \cite{Ghaffari2015} determined that tumor cells continuously change their metabolic pathways to promote proliferation, metastasis and survival, and therefore they constantly depend on the supply of Nutrients. Koo and Yoon \cite{Koo2015} found the extreme up-regulation of Glutaminase-1 and Amino acid transporter-2, proteins related to metabolism in ACC which confirm glutaminolysis( breakdown of glutamine). Also, the expression of glycine, serine, and threonine metabolism in ACC was comparatively higher than in normal tissues and benign cancer. This indicates that tumor cells utilize those nutrients for growth, suggesting that Nutrients help promote malignancy of ACC \citep{Yang2022}. Lyu et al \cite{Serine2025} determined the role of serine synthesis in tumor growth and it was crucial in tumor onset, advancement and acquired immune response. In addition, extreme serine metabolism transforms the tumor as malignant one and promotes oncogenic transformation, immune suppression and therapeutic resistance in tumor cells. Glutamine has a crucial role in the rapid proliferation and metastasis of tumor cells, and high expression of it is associated with high mortality in the case of breast cancer. Also, consumption of glycine is highly correlated with the proliferation rate of cancer cells from the NCI-60 panel cell line \citep{Jain2015}. These findings highlight that the proliferation of tumor cells is influenced by serine and glycine moreover by Nutrients.\\

\noindent 
Furthermore, the availability of Oxygen has been identified as another critical factor influencing tumor progression. Mendonça et al \cite{Mendonça2020} determined that, due to the advancement of ACC the oxygen level declines, causing intratumoural hypoxia in the tumor microenvironment. Similarly, according to \cite{Nejad2021}, Hypoxia occurs due the fast growth of tumor cell. Notably, it can influence angiogenesis (the process of formation of new blood vessels), aggressiveness and metastasis. Moreover, hypoxia helps to accelerate the survival of tumor, immune suppression and makes tumor as therapy resistant. In addition, Kozal and Krzeslak \cite{Kozal2022} described the role of hypoxia in tumor microenvironment which is responsible for immune suppression, altering metabolic pathways, angiogenesis, cell proliferation, metastasis, etc. Thereafter, how tumor cells become adaptive to hypoxia-induced situations and intolerant to chemotherapy and radiation therapy are investigated by Chen et al \citep{Chen2023}. Consequently, to sustain their uncontrolled proliferation, tumor cells experience Oxygen and Nutrient deprivation. At this stage, the hypoxic environment influences the initiation of angiogenesis to meet the demand of an extended blood supply \citep{Yang2024}. According to a statistical analysis provided by Anjum et al \citep{Anjum2022}, among 30 lacrimal gland ACC patents $HIF-1\alpha$ (Hypoxia-Inducible Factor) and $VEGF$ (Vascular Endothelial Growth Factor) expression were $36.6\%$ and $46.6\%$, respectively. Also, these factors are closely associated with perineural invasion and are accountable for poor clinical outcomes. Another study by Chiche et al \citep{Rouleau2010} determined that, if cancer cells remain in a hypoxic environment for a long time ($3$ days), then they become resistant to that chemo which promotes apoptosis.\\

\noindent 
Surgical management is the most common treatment strategy of ACC. Orbital exenteration is one of them, but it is not self-contained. In most of the cases postoperative radiation therapy is required. Because of perineural invasion both the orbit and the skull base must need radiation therapy. Although this complex surgery combined with radiation therapy lead to a low level recurrence. In terms of most of the patents, tumor develops in distant organs like lungs and bones after $1-2$ years and they eventually die from the disease \citep{surgery}. Due to these limitations, eye-sparing surgery has been introduced as an alternative treatment strategy in recent years. Particularly for patients whose tumors are limited to the orbit. In this process, the tumor is surgically removed while preserving the eye. After that adjuvant radiotherapy is applied to reduce the risk of recurrence. Several studies have demonstrated that this treatment strategy provides comparatively more local control of the tumor. Also, the long-term survival rate is quite satisfactory. In many cases, patients remain disease-free for the rest of their lives \citep{eye_sparcing_surgery2018}. Cosse et al \cite{Cosse2008} investigated how hypoxia renders cancer cell chemo-resistant. They described that, when there exists Oxygen deprivation in any cancer-specific area, then there persists an acidic environment. In that case, drug delivery becomes affected. Secondly, hypoxia reduces the effectiveness of chemotherapy because in the absence of Oxygen, free radicals can not be formed. Lastly, hypoxia changes the genomic sequence and structures of growth-related proteins. As a result, there remains no efficacy of chemotherapy.\\

\noindent 
As the present treatment strategies are not so effective and targeted therapies have not even been introduced yet which can control the tumor totally and influence mortality rates. Thus, considering the aggressiveness and limited treatment strategies of the tumor, a newly proposed mathematical model is developed to demonstrate microenvironmental dynamics of Adenoid Cystic Carcinoma in lacrimal gland. This study shows the key biological features of the tumor and feasible treatment intervention.\\

\noindent 
The paper is organized as follows: In section 2, a biologically grounded mathematical model is constructed based on the dynamics of the key changeable factors of the tumor. An analytical analysis of the model is performed on section 3. This section includes boundedness, existence and uniqueness of model solution, equilibrium points, stability analysis, future status of the tumor and sensitivity Analysis. Section 4 contains numerical simulations and their graphical representations. The treatment intervention is introduced in Section 5. Consequently, key findings of the work are presented in Section 6. 
\section{Mathematical Formulation of the Model}
In this section, we develop a newly proposed mathematical model based on concepts and biological mechanisms discussed in previous literature review about the progression of lacrimal gland ACC. The model incorporates the complex molecular interactions of the tumor microenvironment.
\subsection{Assumptions of the Model}
During the formulation of the proposed mathematical model for lacrimal gland ACC, the following assumptions are considered:
\begin{itemize}
\item All state variables, namely Number of Tumor Cells $T(t)$, Oxygen Level $O(t)$, Nutrient Level $N(t)$ and Blood Vessels $ V(t)$ are assumed to be continuous and differentiable functions of time.
\item The dynamics of the tumor microenvironment are studied over a continuous time period of five years($60$ months).
\item Tumor cells continuously consume Oxygen and Nutrients.
\item All parameters are assumed to remain constant throughout the simulation period.
\end{itemize}
\subsection{Model Description}
Lacrimal gland Adenoid Cystic Carcinoma is a rare but highly aggressive malignant tumor. Despite advances in clinical management, the biological mechanisms underlying ACC progression remain poorly understood due to its rarity and the complex interactions between tumor cells and their surrounding microenvironment. In this chapter, we propose a novel mathematical model to investigate the growth dynamics of ACC by explicitly incorporating key components of the tumor microenvironment.\\

\noindent 
The growth of the tumor is associated with changes in the tumor microenvironment. To demonstrate the situation, we denote the number of Tumor cells as $T(t)$, $O(t)$ is expressed as the Oxygen level also Nutrient level is $N(t)$ and Tumor-induced Abnormal Blood Vessels is denoted by $V(t)$ at any time $t$. The number of tumor cells naturally grows logistically. On the other hand, tumor cells are decreased for immunosuppression $\sigma T,$ and natural apoptosis $\mu T$. Also the presence of Oxygen and nutrients influences the growth of tumor cells.
\begin{equation}
	\frac {dT}{dt}= \alpha T\left(1-\frac{T}{K}\right)+\beta OT+\gamma NT-\sigma T -\mu T
	\label{T'}
\end{equation}
The Nutrient level is escalated by the nutrition supply of stromal-cell $\Lambda_n$ in the tumor. The increase in abnormal blood vessels stimulates the level of nutrients at a constant rate $\gamma_1 $. However, by the influence of hyperplasia Nutrient level drops by $\eta_1 TN$ and the natural diffusion rate of nutrients is $\mu_1$.
\begin{equation}
	\frac {dN}{dt}= \Lambda_n +\gamma_1 VN -\eta_1 T N-\mu_1N
	\label{N'}
\end{equation}
The rise in oxygen level also depends on the increase in  abnormal blood Vessels. On the other hand, extended tumor cells accelerate oxygen consumption which create hypoxia situation. Also, oxygen is naturally diffused at a rate $\mu_2$.
\begin{equation}
	\frac {dO}{dt}=\Lambda_O+\beta_1 VO-\eta_2 T O-\mu_2O
	\label{O'}
\end{equation}
Angiogenesis, formation of irregular vascular structure, is enhanced by tumor-induced hypoxia and nutrient availability, while blood vessels decay due to a natural regression rate $\sigma_1$
\begin{equation}
	\frac {dV}{dt}=\delta(\frac{K_h} {K_h+O}) T+\epsilon (\frac{N}{K_n+N})V-\sigma_1 V
	\label{V'}
\end{equation}
The descriptions of the model parameters used in the proposed system are presented in the following table.
\begin{center}
  \begin{table}[hbt!]
   \caption{Parameter value and their descriptions}

    \begin{tabular}{p{1.3cm} p{6.0cm} p{6.5cm}}
      \hline
      \textbf{Symbol} & \textbf{Parameter's Description} & \textbf{Value}\\
      \hline    
$\alpha$ &  Natural growth rate of tumor cell  & $0.35 month^{-1}$     \\ 
$\beta$    & Tumor cell enhancement rate due to Oxygen  & $0.0001 (mmHg)^{-1} month^{-1}$      \\
$\gamma$   & Tumor cell enhancement rate due to Nutrients & $0.001 month^{-1}mol^{-1}m^3$      \\
$\sigma$  & Rate of immunosuppression  & $0.016 month^{-1}$   \\
$\mu$    &  Natural apoptosis rate of tumor cell    & $0.082month^{-1}$ \\
$\Lambda_n$   & Rate of Nutrient supply by stromal cell   &  $0.5 molm^{-3}month^{-1}$        \\
$\gamma_1$  &  Rate of Nutrient supply by abnormal blood vessels & $0.0005 vessels^{-1}mm^2month^{-1}$     \\
$\eta_1$  
 & Nutrient consumption rate by tumor cells  &  $5 \times 10^{-11}cells^{-1}month^{-1}$     \\
 $\mu_1$    &  Natural regression rate of Nutrients   & $0.05month^{-1}$ \\
 $\Lambda_o$   & Rate of Oxygen supply by healthy vessels &  $2.0 mmHgmonth^{-1}$        \\
$\beta_1$   & Rate of Oxygen supply via  abnormal vessels & $0.001vessels^{-1}mm^2month^{-1}$    \\
$\eta_2$    & Oxygen  consumption rate by tumor cells & $6 \times 10^{-11}cells^{-1}month^{-1}$ \\
$\mu_2$    &  Natural regression rate of Oxygen & $0.04month^{-1}$ \\
$ \delta$  & Angiogenesis stimulation rate by tumor cells &  $5 \times 10^{-11} vesselsmm^{-2} cells^{-1}month^{-1}$\\
$\sigma_1$  & Natural decay rate of vessels & $0.0015 month^{-1}$  \\
$\epsilon$  &  Nutrient induced vascular growth rate  &  $0.01 month^{-1}$        \\

$K$   & Carrying capacity (maximum sustainable tumor cell due to limitation of space and resource) & $12 \times 10^{9} cells$\\
$K_n$   & Saturation constant for nutrient-dependent vascular response& $2molm^{-3}$\\
$K_h$   & Saturation constant for oxygen-dependent vascular response & $40mmHg$\\
      \hline
    \end{tabular}
    \label{table:paramters value}
    \end{table}
  \end{center}

\section{Qualitative Analysis of the Model}
This section demonstrates the qualitative analysis of the newly proposed model, focusing on the 
existence of the system, stability at equilibrium points, different dynamic behaviors, etc.
\subsection{Positivity of State Variables}
\begin{theorem}
   \label{Theorem:Positivity Theorem}\index{Positivity}
Considering $T(0)\ge 0, N(0)\ge0, O(0)\ge 0$ and $V(0)\ge 0$, then $T(t), N(t), O(t)$ and $V(t)$ will be always positive for all $t \in [0,T_f]$ in $\mathbb{R}^{+}_4$ where $T_f>0$. 
\end{theorem}
\begin{proof}
    Taking all parameters of the system and all initial values to be positive, we establish the non-negativity of the state variables through differential inequalities.
    
\noindent From equation \eqref{T'}, the dynamics of the tumor cells is given by
\begin{align*}
	\frac {dT}{dt} &= \alpha T\left(1-\frac{T}{K}\right)+\beta OT+\gamma NT-\sigma T -\mu T \\
	\frac {dT}{dt} &= \left[ \alpha - \frac{\alpha T}{K} + \beta O + \gamma N - \sigma - \mu \right] T \\
	\Rightarrow \frac{dT}{T} &= \left[ (\alpha - \sigma - \mu) - \frac{\alpha T}{K} + \beta O + \gamma N \right] dt
\end{align*}
Integrating both sides from $0$ to $t$, we get,
\begin{align*}
	\ln \left(\frac{T(t)}{T(0)}\right) &= \int_{0}^{t} \left[ (\alpha - \sigma - \mu) - \frac{\alpha T(\tau)}{K} + \beta O(\tau) + \gamma N(\tau) \right] d\tau \\
	\Rightarrow T(t) &= T(0) e^{ (\alpha - \sigma - \mu)t - \int_{0}^{t}\frac{\alpha T(\tau)}{K}d\tau + \int_{0}^{t}\beta O(\tau)d\tau + \int_{0}^{t}\gamma N(\tau)d\tau } > 0
\end{align*}
Since $T(t) > 0$, we can consider $\eta_1 T = \eta_0$ is always positive.

\medskip
\noindent Next, from equation \eqref{N'}, the nutrient level equation is
\begin{align*}
	\frac {dN}{dt} &= \Lambda_n +\gamma_1 VN -\eta_1 T N-\mu_1N \\
	\frac {dN}{dt} &= \Lambda_n + \left( \gamma_1 V - \eta_0 - \mu_1 \right) N 
\end{align*}
Since $\Lambda_n > 0$, we obtain the following inequality:
\begin{align*}
	\frac {dN}{dt} &\geq \left( \gamma_1 V - \eta_0 - \mu_1 \right) N \\
	\Rightarrow \frac{dN}{N} &\geq \left( \gamma_1 V - \eta_0 - \mu_1 \right) dt
\end{align*}
Integrating both sides from $0$ to $t$ yields,
\begin{align*}
	\ln \left(\frac{N(t)}{N(0)}\right) &\geq \int_{0}^{t} \left( \gamma_1 V(\tau) - \eta_0 - \mu_1 \right) d\tau \\
	\Rightarrow N(t) &\geq N(0) e^{ -(\eta_0 + \mu_1)t + \int_{0}^{t} \gamma_1 V(\tau) d\tau } > 0
\end{align*}

\medskip
\noindent Similarly, from equation \eqref{O'}, the oxygen level dynamics can be expressed by
\begin{align*}
	\frac {dO}{dt} &= \Lambda_O + \beta_1 VO - \eta_2 T O - \mu_2 O \\
	\frac {dO}{dt} &= \Lambda_O + \left( \beta_1 V - \eta_O - \mu_2 \right) O
\end{align*}
Since $\Lambda_O > 0$, dropping the source term gives:
\begin{align*}
	\frac {dO}{dt} &\geq \left( \beta_1 V - \eta_O - \mu_2 \right) O \\
	\Rightarrow \frac{dO}{O} &\geq \left( \beta_1 V - \eta_O - \mu_2 \right) dt
\end{align*}
Integrating both sides from $0$ to $t$, we have,
\begin{align*}
	\ln \left(\frac{O(t)}{O(0)}\right) &\geq \int_{0}^{t} \left( \beta_1 V(\tau) - \eta_O - \mu_2 \right) d\tau \\
	\Rightarrow O(t) &\geq O(0) e^{ -(\eta_O + \mu_2)t + \int_{0}^{t} \beta_1 V(\tau) d\tau } > 0
\end{align*}

\medskip
\noindent Finally, from equation \eqref{V'},
\begin{align*}
	\frac {dV}{dt} &= \delta\left(\frac{K_h} {K_h+O}\right) T+\epsilon \left(\frac{N}{K_n+N}\right)V-\sigma_1 V \\
	\frac {dV}{dt} &= \delta\left(\frac{K_h} {K_h+O}\right) T + \left[ \epsilon \left(\frac{N}{K_n+N}\right) - \sigma_1 \right] V
\end{align*}
Since, $\delta, K_h, O, T > 0$, which leads to
\begin{align*}
	\frac {dV}{dt} &\geq \left[ \epsilon \left(\frac{N}{K_n+N}\right) - \sigma_1 \right] V \\
	\Rightarrow \frac{dV}{V} &\geq \left[ \epsilon \left(\frac{N}{K_n+N}\right) - \sigma_1 \right] dt
\end{align*}
Integrating both sides from $0$ to $t$, we get,
\begin{align*}
	\ln \left(\frac{V(t)}{V(0)}\right) &\geq \int_{0}^{t} \left[ \epsilon \left(\frac{N(\tau)}{K_n+N(\tau)}\right) - \sigma_1 \right] d\tau \\
	\Rightarrow V(t) &\geq V(0) e^{ -\sigma_1 t + \int_{0}^{t} \epsilon \left(\frac{N(\tau)}{K_n+N(\tau)}\right) d\tau} > 0
\end{align*}
Therefore, all the state variables of the proposed model are strictly non-negative for all $t > 0$.
\end{proof}
\subsection{Boundedness of State Variables}
\begin{theorem}
   \label{Theorem:Boundedness Theorem}\index{Boundedness}
The solution trajectories $(T(t), N(t), O(t), V(t))$ of the system \eqref{T'}--\eqref{V'} entering from non-negative initial conditions remain uniformly bounded within a compact region $\Omega$.
\end{theorem}

\begin{proof}
Let us define $Z(t)$ as
\begin{equation}
    Z(t) = T(t) + \frac{\gamma}{\eta_1} N(t) + \frac{\beta}{\eta_2} O(t)
\end{equation}
In the absence of abnormal blood vessels, differentiating $Z(t)$ with respect to time $t$ we get,
\begin{equation}
    \frac{dZ}{dt} = \frac{dT}{dt} + \frac{\gamma}{\eta_1} \frac{dN}{dt} + \frac{\beta}{\eta_2} \frac{dO}{dt}
\end{equation}
Substituting the corresponding equations, the quadratic bilinear resource consumption terms ($-\gamma NT$ and $-\beta OT$) perfectly and identically cancel out with their respective metabolic uptake rates.
\begin{equation}
    \therefore\frac{dZ}{dt} = \alpha T - \frac{\alpha T^2}{K} - (\sigma + \mu)T + \frac{\gamma}{\eta_1}\Lambda_n - \frac{\gamma \mu_1}{\eta_1}N + \frac{\beta}{\eta_2}\Lambda_O - \frac{\beta \mu_2}{\eta_2}O
\end{equation}

\begin{equation}
    \implies\frac{dZ}{dt} + \lambda_0 Z \leq \left[ (\alpha - \sigma - \mu + \lambda_0)T - \frac{\alpha T^2}{K} \right] - \frac{\gamma}{\eta_1}(\mu_1 - \lambda_0)N - \frac{\beta}{\eta_2}(\mu_2 - \lambda_0)O + \frac{\gamma}{\eta_1}\Lambda_n + \frac{\beta}{\eta_2}\Lambda_O
    \label{eq:dZ_sequential}
\end{equation}
Now the maximum value of $\left( -\frac{\alpha T^2}{K} + (\alpha - \sigma - \mu + \lambda_0)T \right)$ is $ \frac{K}{2\alpha}\left(\alpha - \sigma - \mu + \lambda_0\right)^2$ 
Consequently, the entire right-hand side is bounded above by a  positive constant $\Theta$. So,
\begin{equation}
    \frac{dZ}{dt} + \lambda_0 Z \leq \Theta_0
\end{equation}
Applying Gronwall's comparison theorem, we find $\limsup_{t \to \infty} Z(t) \leq \frac{\Theta}{\lambda_0}$. So, this establishes that $T(t)$, $N(t)$, and $O(t)$ are bounded. \\
Now, let the upper bound for the tumor cells be denoted as $\limsup_{t \to \infty} T(t) \leq M_T$.

Next, we evaluate the abnormal blood vessel dynamics from equation \eqref{V'}. Utilizing the natural boundary properties of the saturation terms ($\frac{K_h}{K_h+O} \leq 1$ and $\frac{N}{K_n+N} \leq 1$), and replacing $T(t)$ with its upper bound $M_T$, the differential structure reduces to
\begin{align}
    \frac{dV}{dt} \leq \delta(1) M_T + \epsilon(1)V - \sigma_1 V \nonumber \\
    \Rightarrow \frac{dV}{dt} +(\sigma_1 - \epsilon)V \leq \delta M_T
    \label{eq:V_final_bound}
\end{align}

Considering the limit supremum leads us to
\begin{equation}
    \limsup_{t \to \infty} V(t) \leq \frac{\delta M_T}{\sigma_1 - \epsilon} 
\end{equation}
Thus, all the state variables of the model are bounded.
\end{proof}

\subsection{Existence and Uniqueness of the Model Solution}

\begin{theorem}
   \label{Theorem:Existence and Uniqueness}\index{Existence and Uniqueness} Let $D$ be the domain defined in such a way that Lipschitz conditions are satisfied. Then, for all non-negative initial conditions, the solution of the system exists and is also unique at the same time for all time $t\ge0$ in the domain $D$.
\end{theorem}

\begin{proof}
According to [ ], we need to show that
\begin{align*}
    \frac{\partial f_i}{\partial x_j}, \quad i,j = 1,2,3,4
\end{align*}
is continuous and bounded in the domain $D$. So, from equation \eqref{T'}, we can write
\begin{align}
    f_1 = \frac{dT}{dt} = \alpha T\left(1-\frac{T}{K}\right)+\beta OT+\gamma NT-\sigma T -\mu T.
    \label{eq:f1}
\end{align}

Differentiating $f_1$ with respect to $T$, we have
\begin{align*}
    \frac{\partial f_1}{\partial T} = \alpha-\frac{2\alpha T}{K}+\beta O+\gamma N-\sigma-\mu.
\end{align*}

Therefore,
\begin{align*}
    \left|\frac{\partial f_1}{\partial T}\right| = \left|\alpha-\frac{2\alpha T}{K}+\beta O+\gamma N-\sigma-\mu\right| < \infty.
\end{align*}

Again, differentiating $f_1$ with respect to $N$, $O$, and $V$ respectively, we get
\begin{align*}
    \frac{\partial f_1}{\partial N} = \gamma T \quad &\therefore \left|\frac{\partial f_1}{\partial N}\right| < \infty, \\
    \frac{\partial f_1}{\partial O} = \beta T \quad &\therefore \left|\frac{\partial f_1}{\partial O}\right| < \infty, \\
    \frac{\partial f_1}{\partial V} = 0 \quad &\therefore \left|\frac{\partial f_1}{\partial V}\right| < \infty.
\end{align*}

Also from equation \eqref{N'}, we can write
\begin{align}
    f_2 = \frac{dN}{dt} = \Lambda_n +\gamma_1 VN -\eta_1 T N-\mu_1N.
    \label{eq:f2}
\end{align}

Differentiating $f_2$ with respect to $T$, $N$, $O$, and $V$ respectively, we get
\begin{align*}
    \frac{\partial f_2}{\partial T} = -\eta_1 N \quad &\therefore \left|\frac{\partial f_2}{\partial T}\right| < \infty, \\
    \frac{\partial f_2}{\partial N} = \gamma_1 V-\eta_1 T-\mu_1 \quad &\therefore \left|\frac{\partial f_2}{\partial N}\right| < \infty, \\
    \frac{\partial f_2}{\partial O} = 0 \quad &\therefore \left|\frac{\partial f_2}{\partial O}\right| < \infty, \\
    \frac{\partial f_2}{\partial V} = \gamma_1 N \quad &\therefore \left|\frac{\partial f_2}{\partial V}\right| < \infty.
\end{align*}

Furthermore, from equation \eqref{O'}, we can write
\begin{align}
    f_3 = \frac{dO}{dt} = \Lambda_O+\beta_1 VO-\eta_2 T O-\mu_2O.
    \label{eq:f3}
\end{align}

Differentiating $f_3$ with respect to $T$, $N$, $O$, and $V$ respectively, we get
\begin{align*}
    \frac{\partial f_3}{\partial T} = -\eta_2 O \quad &\therefore \left|\frac{\partial f_3}{\partial T}\right| < \infty, \\
    \frac{\partial f_3}{\partial N} = 0 \quad &\therefore \left|\frac{\partial f_3}{\partial N}\right| < \infty, \\
    \frac{\partial f_3}{\partial O} = \beta_1 V-\eta_2 T-\mu_2 \quad &\therefore \left|\frac{\partial f_3}{\partial O}\right| < \infty, \\
    \frac{\partial f_3}{\partial V} = \beta_1 O \quad &\therefore \left|\frac{\partial f_3}{\partial V}\right| < \infty.
\end{align*}

And from equation \eqref{V'}, we can write
\begin{align}
    f_4 = \frac{dV}{dt} = \delta\left(\frac{K_h}{K_h+O}\right)T+\epsilon\left(\frac{N}{K_n+N}\right)V-\sigma_1 V.
    \label{eq:f4}
\end{align}

Differentiating $f_4$ with respect to $T$, $N$, $O$, and $V$ respectively, we get
\begin{align*}
    \frac{\partial f_4}{\partial T} = \delta\left(\frac{K_h}{K_h+O}\right) \quad \therefore \left|\frac{\partial f_4}{\partial T}\right| < \infty,
\end{align*}
\begin{align*}
    \frac{\partial f_4}{\partial O} = -\delta T\left[\frac{K_h}{(K_h+O)^2}\right] \quad \therefore \left|\frac{\partial f_4}{\partial O}\right| < \infty,
\end{align*}
\begin{align*}
    \frac{\partial f_4}{\partial N} = \epsilon V\left[\frac{K_n}{(K_n+N)^2}\right] \quad \therefore \left|\frac{\partial f_4}{\partial N}\right| < \infty,
\end{align*}
\begin{align*}
    \frac{\partial f_4}{\partial V} = \epsilon\left(\frac{N}{K_n+N}\right)-\sigma_1 \quad \therefore \left|\frac{\partial f_4}{\partial V}\right| < \infty,
\end{align*}

Thus, all the partial derivatives
\begin{align*}
    \left|\frac{\partial f_i}{\partial T}\right|,\ \left|\frac{\partial f_i}{\partial N}\right|,\ \left|\frac{\partial f_i}{\partial O}\right|,\ \left|\frac{\partial f_i}{\partial V}\right| < \infty, \qquad i=1,2,3,4,
\end{align*}
are bounded and continuous in $D$, satisfying the Lipschitz condition. Therefore, a unique solution to the proposed system exists in the region $D$.
\end{proof}
\subsection{Equilibrium Analysis}
\subsubsection{Tumor-Free Equilibrium (TFE)}
To analyze and understand the system behavior, steady state analysis is the first step. The equilibrium state represents a stage at which the interactions among tumor cells, oxygen, nutrients and other microenvironmental factors attain a dynamic balance causing the system variables to remain unchanged over time.
An equilibrium point of the system is obtained by setting all the rate of change equations to zero, i.e., 
\begin{align*}
    \frac{dT}{dt} = 0, \quad \frac{dN}{dt} = 0, \quad \frac{dO}{dt} = 0, \quad \frac{dV}{dt} = 0
\end{align*}
To investigate the biological state where the tumor is completely eradicated, we assume the absence of tumor cells, setting $T = 0$. Let the corresponding tumor-free state be denoted as $E_0 = (T_0, N_0, O_0, V_0) = (0, N_0, O_0, V_0)$.

\noindent Substituting $T = 0$ into equation \eqref{V'}, the dynamics of the vascular network simplifies to
\begin{align*}
    \left[ \epsilon \left(\frac{N_0}{K_n+N_0}\right) - \sigma_1 \right] V_0 = 0
\end{align*}
In a healthy biological state without tumor-induced angiogenic signals, abnormal blood vessels do not proliferate, yielding $V_0 = 0$.

\noindent Next, substituting $T = 0$ and $V_0 = 0$ into the nutrient equation \eqref{N'}, we get,
\begin{align*}
    \Lambda_n - \mu_1 N_0 = 0 \implies N_0 = \frac{\Lambda_n}{\mu_1}
\end{align*}
Similarly, substituting $T = 0$ and $V_0 = 0$ into the oxygen dynamics equation \eqref{O'} yields
\begin{align*}
    \Lambda_O - \mu_2 O_0 = 0 \implies O_0 = \frac{\Lambda_O}{\mu_2}
\end{align*}
Thus, the system admits a unique \textbf{Tumor-Free Equilibrium (TFE)} point, denoted by
\begin{equation}
    E_0 = \left( 0, \frac{\Lambda_n}{\mu_1}, \frac{\Lambda_O}{\mu_2}, 0 \right)
    \label{eq:TFE_point}
\end{equation}
\subsubsection{Endemic Equilibrium}
The endemic equilibrium point is where tumor cell populations and other microenvironmental factors coexist continuously ($T^* > 0, N^* > 0, O^* > 0, V^* > 0$). Setting the right-hand side of equations \eqref{T'}--\eqref{V'} to zero, endemic equilibrum $E_1(T^*, N^*, O^*, V^*)$ can be obtained by solving the following relations
\begin{align}
    T^* &= K \left[ 1 + \frac{\beta O^* + \gamma N^* - (\sigma + \mu)}{\alpha} \right] \label{eq:T_star} \\
    V^* &= \frac{\delta K_h T^*}{(K_h + O^*) \left[ \sigma_1 - \epsilon \left(\frac{N^*}{K_n+N^*}\right) \right]} \label{eq:V_star} \\
    N^* &= \frac{\Lambda_n}{\mu_1 + \eta_1 T^* - \gamma_1 V^*} \label{eq:N_star} \\
    O^* &= \frac{\Lambda_O}{\mu_2 + \eta_2 T^* - \beta_1 V^*} \label{eq:O_star}
\end{align}
\subsection {Future Status of the Tumor}
The reproduction number $R_0$ is a threshold indicator of determining the stability of the disease-free state. By using the next-generation matrix approach, we have determined the basic reproduction number. The new infection matrix and its transfer matrix can be denoted by $F$ and $V$, respectively.

\begin{align*}
    F = \begin{pmatrix} \alpha + \beta O_0 + \gamma N_0 \end{pmatrix}, \qquad
    V = \begin{pmatrix} \sigma + \mu \end{pmatrix}
\end{align*}

The basic reproduction number $R_0$ is the spectral radius of the next-generation matrix $FV^{-1}$. Therefore, the reproduction number of our model is
\begin{align}
    R_0 = \frac{\alpha + \beta O_0 + \gamma N_0}{\sigma + \mu} = \frac{\alpha + \beta \left(\frac{\Lambda_O}{\mu_2}\right) + \gamma \left(\frac{\Lambda_n}{\mu_1}\right)}{\sigma + \mu}.
\end{align}
This can determine the further status of the tumor, whether the tumor will die out or persist in the microenvironment. Mathematically, if $R_0<1$ then it represents that the tumor will be eradicated and if $R_0>1$ then this indicates the tumor will persist.
\subsection{Local Stability Analysis}
\begin{theorem}
   \label{Theorem:Stability}\index{Local Stability at Disease free equilibrium} The model is locally asymptotically stable at the disease-free equilibrium point if $R_0<1$ and $\epsilon \left(\frac{N_0}{K_n+N_0}\right) < \sigma_1$, otherwise unstable.
\end{theorem}
\begin{proof}
Let us consider the Jacobian matrix $J$. Then, at the disease-free equilibrium point $J$ can be written as
\begin{align*}
    J_0 = \begin{pmatrix}
    (\alpha + \beta O_0 + \gamma N_0 - \sigma - \mu) & 0 & 0 & 0 \\
    -\eta_1 N_0 & -\mu_1 & 0 & \gamma_1 N_0 \\
    -\eta_2 O_0 & 0 & -\mu_2 & \beta_1 O_0 \\
    \delta \left(\frac{K_h}{K_h+O_0}\right) & 0 & 0 & \epsilon \left(\frac{N_0}{K_n+N_0}\right) - \sigma_1
    \end{pmatrix}.
\end{align*}

Then, the characteristic equation of the Jacobian matrix is $\begin{vmatrix}
    (J_0-\lambda I)
\end{vmatrix}=0$. Hence,
\begin{align*}
    \begin{vmatrix}
    (\alpha + \beta O_0 + \gamma N_0 - \sigma - \mu) - \lambda & 0 & 0 & 0 \\
    -\eta_1 N_0 & -\mu_1 - \lambda & 0 & \gamma_1 N_0 \\
    -\eta_2 O_0 & 0 & -\mu_2 - \lambda & \beta_1 O_0 \\
    \delta \left(\frac{K_h}{K_h+O_0}\right) & 0 & 0 & \left[\epsilon \left(\frac{N_0}{K_n+N_0}\right) - \sigma_1\right] - \lambda
    \end{vmatrix} = 0.
\end{align*}

Expanding along the first row, the first eigenvalue is directly obtained as $\lambda_1 = (\alpha + \beta O_0 + \gamma N_0) - (\sigma + \mu)$, and the remaining determinant reduces to the following form:
\begin{align*}
    \left[(\alpha + \beta O_0 + \gamma N_0 - \sigma - \mu) - \lambda\right]
    \begin{vmatrix}
    -\mu_1 - \lambda & 0 & \gamma_1 N_0 \\
    0 & -\mu_2 - \lambda & \beta_1 O_0 \\
    0 & 0 & \left[\epsilon \left(\frac{N_0}{K_n+N_0}\right) - \sigma_1\right] - \lambda
    \end{vmatrix} = 0.
\end{align*}

Since this reduced $3 \times 3$ determinant is in a clear upper-triangular layout (as all entries below the main diagonal are zero), it simply reduces to the product of its diagonal entries:
\begin{align*}
    \left[(\alpha + \beta O_0 + \gamma N_0 - \sigma - \mu) - \lambda\right] [-\mu_1 - \lambda] [-\mu_2 - \lambda] \left[\epsilon \left(\frac{N_0}{K_n+N_0}\right) - \sigma_1 - \lambda\right] = 0.
\end{align*}

Solving this characteristic equation gives the four distinct eigenvalues of the system:
\begin{align*}
    \lambda_1 = (\alpha + \beta O_0 + \gamma N_0) - (\sigma + \mu), \qquad \lambda_2 = -\mu_1, \qquad \lambda_3 = -\mu_2, \qquad \lambda_4 = \epsilon \left(\frac{N_0}{K_n+N_0}\right) - \sigma_1.
\end{align*}

It is clear that $\lambda_2$ and $\lambda_3$ are intrinsically negative. The eigenvalue $\lambda_4$ is negative by the structural vascular constraint $\epsilon \left(\frac{N_0}{K_n+N_0}\right) < \sigma_1$. The remaining eigenvalue $\lambda_1$ is negative if and only if
\begin{align*}
    &\alpha + \beta O_0 + \gamma N_0 < \sigma+\mu \\
    &\implies\frac{\alpha + \beta O_0 + \gamma N_0}{\sigma+\mu} < 1 \\
    & \implies R_0 < 1.
\end{align*}

Therefore, all the eigenvalues of $J_0$ have negative real parts if and only if $R_0<1$ and the vascular regression condition holds. As a result, the model is locally asymptotically stable at the disease-free equilibrium point.
\end{proof}

\subsection{Local Asymptotic Stability of the Endemic Equilibrium}
To analyze the local asymptotic stability of the endemic equilibrium state $E_1(T^*, N^*, O^*, V^*)$, we evaluate the linearized system around $E_1$. The general Jacobian matrix $J$ of the system \eqref{eq:f1}--\eqref{eq:f4} is defined by
\begin{equation}
    J = \begin{pmatrix}  
        \frac{\partial f_1}{\partial T} & \frac{\partial f_1}{\partial N} & \frac{\partial f_1}{\partial O} & \frac{\partial f_1}{\partial V} \\[1ex]
        \frac{\partial f_2}{\partial T} & \frac{\partial f_2}{\partial N} & \frac{\partial f_2}{\partial O} & \frac{\partial f_2}{\partial V} \\[1ex]
        \frac{\partial f_3}{\partial T} & \frac{\partial f_3}{\partial N} & \frac{\partial f_3}{\partial O} & \frac{\partial f_3}{\partial V} \\[1ex]
        \frac{\partial f_4}{\partial T} & \frac{\partial f_4}{\partial N} & \frac{\partial f_4}{\partial O} & \frac{\partial f_4}{\partial V}  
    \end{pmatrix}
    \label{eq:Jacobian_generic}
\end{equation}

Now, we have,
\begin{align*}
    \frac{\partial f_1}{\partial T} &= \alpha \left(1 - \frac{2T}{K}\right) + \beta O + \gamma N - \sigma - \mu, \quad 
    \frac{\partial f_1}{\partial N} = \gamma T, \quad \frac{\partial f_1}{\partial O} = \beta T, \quad \frac{\partial f_1}{\partial V} = 0, \\
    \frac{\partial f_2}{\partial T} &= -\eta_1 N, \quad \frac{\partial f_2}{\partial N} = \gamma_1 V - \eta_1 T - \mu_1, \quad 
    \frac{\partial f_2}{\partial O} = 0, \quad \frac{\partial f_2}{\partial V} = \gamma_1 N, \\
    \frac{\partial f_3}{\partial T} &= -\eta_2 O, \quad \frac{\partial f_3}{\partial N} = 0, \quad 
    \frac{\partial f_3}{\partial O} = \beta_1 V - \eta_2 T - \mu_2, \quad \frac{\partial f_3}{\partial V} = \beta_1 O, \\
    \frac{\partial f_4}{\partial T} &= \delta \left(\frac{K_h}{K_h+O}\right), \quad 
    \frac{\partial f_4}{\partial N} = \frac{\epsilon V K_n}{(K_n+N)^2}, \quad 
    \frac{\partial f_4}{\partial O} = -\frac{\delta T K_h}{(K_h+O)^2}, \quad 
    \frac{\partial f_4}{\partial V} = \epsilon \left(\frac{N}{K_n+N}\right) - \sigma_1
\end{align*}

To simplify the diagonal elements of the Jacobian matrix at the endemic equilibrium $E_1$, we invoke the steady-state algebraic relations satisfying $\frac{dT}{dt}=\frac{dN}{dt}=\frac{dO}{dt}=\frac{dV}{dt}=0$.
\begin{align}
    \alpha \left(1-\frac{T^*}{K}\right) + \beta O^* + \gamma N^* - \sigma - \mu &= 0 \implies \alpha + \beta O^* + \gamma N^* - \sigma - \mu = \frac{\alpha T^*}{K}, \label{eq:sub1} \\
    \Lambda_n + \gamma_1 V^* N^* - \eta_1 T^* N^* - \mu_1 N^* &= 0 \implies \gamma_1 V^* - \eta_1 T^* - \mu_1 = -\frac{\Lambda_n}{N^*}, \label{eq:sub2} \\
    \Lambda_O + \beta_1 V^* O^* - \eta_2 T^* O^* - \mu_2 O^* &= 0 \implies \beta_1 V^* - \eta_2 T^* - \mu_2 = -\frac{\Lambda_O}{O^*}, \label{eq:sub3} \\
    \delta\left(\frac{K_h} {K_h+O^*}\right) T^* + \epsilon \left(\frac{N^*}{K_n+N^*}\right)V^* - \sigma_1 V^* &= 0 \implies \epsilon \left(\frac{N^*}{K_n+N^*}\right) - \sigma_1 = -\frac{\delta T^*}{V^*} \left(\frac{K_h}{K_h+O^*}\right) \label{eq:sub4}
\end{align}

Now, substituting these identities \eqref{eq:sub1}--\eqref{eq:sub4} into each corresponding diagonal component of the Jacobian matrix evaluated at $E_1$,
\begin{align*}
    \frac{\partial f_1}{\partial T}\Bigg|_{E^*} &= \alpha - \frac{2\alpha T^*}{K} + \beta O^* + \gamma N^* - \sigma - \mu = \left(\alpha + \beta O^* + \gamma N^* - \sigma - \mu\right) - \frac{2\alpha T^*}{K} = \frac{\alpha T^*}{K} - \frac{2\alpha T^*}{K} = -\frac{\alpha T^*}{K}, \\
    \frac{\partial f_2}{\partial N}\Bigg|_{E^*} &= \gamma_1 V^* - \eta_1 T^* - \mu_1 = -\frac{\Lambda_n}{N^*}, \\
    \frac{\partial f_3}{\partial O}\Bigg|_{E^*} &= \beta_1 V^* - \eta_2 T^* - \mu_2 = -\frac{\Lambda_O}{O^*}, \\
    \frac{\partial f_4}{\partial V}\Bigg|_{E^*} &= \epsilon \left(\frac{N^*}{K_n+N^*}\right) - \sigma_1 = -\frac{\delta T^*}{V^*} \left(\frac{K_h}{K_h+O^*}\right)
\end{align*}

Applying these reductions, we get,
\begin{equation}
    J(E_1) = \begin{pmatrix}  
        -a_{11} & a_{12} & a_{13} & 0 \\  
        -a_{21} & -a_{22} & 0 & a_{24} \\  
        -a_{31} & 0 & -a_{33} & a_{34} \\  
        a_{41} & a_{42} & -a_{43} & -a_{44}  
    \end{pmatrix}
    \label{eq:Jacobian_matrix}
\end{equation}
where,
\begin{align*}
    a_{11} &= \frac{\alpha T^*}{K}, \quad a_{12} = \gamma T^*, \quad a_{13} = \beta T^*, \\
    a_{21} &= \eta_1 N^*, \quad a_{22} = \frac{\Lambda_n}{N^*}, \quad a_{24} = \gamma_1 N^*, \\
    a_{31} &= \eta_2 O^*, \quad a_{33} = \frac{\Lambda_O}{O^*}, \quad a_{34} = \beta_1 O^*, \\
    a_{41} &= \delta\left(\frac{K_h} {K_h+O^*}\right), \quad a_{42} = \frac{\epsilon V^* K_n}{(K_n+N^*)^2}, \\
    a_{43} &= \frac{\delta T^* K_h}{(K_h+O^*)^2}, \quad a_{44} = \frac{\delta T^*}{V^*}\left(\frac{K_h} {K_h+O^*}\right)
\end{align*}

The characteristic equation of the Jacobian \eqref{eq:Jacobian_matrix} is
\begin{equation}
    \lambda^4 + A_1 \lambda^3 + A_2 \lambda^2 + A_3 \lambda + A_4 = 0
    \label{eq:char_polynomial}
\end{equation}
So, the relation between the components $a_{ij}$ and polynomial coefficients is defined as
\begin{align*}
    A_1 = & \, a_{11} + a_{22} + a_{33} + a_{44}, \\
    A_2 = & \, (a_{11}a_{22} + a_{12}a_{21}) + (a_{11}a_{33} + a_{13}a_{31}) + a_{11}a_{44} + a_{22}a_{33} + (a_{22}a_{44} - a_{24}a_{42}) + (a_{33}a_{44} + a_{34}a_{43}), \\
    A_3 = & \, a_{22}a_{33}a_{44} + a_{22}a_{34}a_{43} - a_{24}a_{33}a_{42} + a_{11}a_{33}a_{44} + a_{11}a_{34}a_{43} + a_{13}a_{31}a_{44} - a_{13}a_{34}a_{41} \\
          & + a_{11}a_{22}a_{44} - a_{11}a_{24}a_{42} + a_{12}a_{21}a_{44} - a_{12}a_{24}a_{41} + a_{11}a_{22}a_{33} + a_{12}a_{21}a_{33} + a_{13}a_{22}a_{31}, \\
    A_4 = & \, a_{11}a_{22}a_{33}a_{44} + a_{11}a_{22}a_{34}a_{43} - a_{11}a_{24}a_{33}a_{42} + a_{12}a_{21}a_{33}a_{44} + a_{12}a_{21}a_{34}a_{43} \\
          & - a_{12}a_{24}a_{31}a_{43} - a_{12}a_{24}a_{33}a_{41} + a_{13}a_{21}a_{34}a_{42} + a_{13}a_{22}a_{31}a_{44} - a_{13}a_{22}a_{34}a_{41} - a_{13}a_{24}a_{31}a_{42}
\end{align*}

According to the Routh-Hurwitz criteria, the endemic equilibrium point $E^*$ is locally asymptotically stable if the following conditions are satisfied. We have shown that the first condition is already satisfied.
\begin{equation*}
    A_1 > 0, \quad A_4 > 0, \quad A_1 A_2 - A_3 > 0, \quad \text{and} \quad A_1 A_2 A_3 - A_3^2 - A_1^2 A_4 > 0
\end{equation*}
\subsection{Global Asymptotic Stability at tumor-free equilibrium point}

For our model, the global stability at the tumor-free point is analy using the method applied in \citep{Amir2020}. The method is summarized as, the proposed model is reduced into the following two subsystems given by
\begin{align}
    \frac{\mathrm{d}X_1}{\mathrm{d}t} &= G(X_1, X_2), \nonumber \\
    \frac{\mathrm{d}X_2}{\mathrm{d}t} &= H(X_1, X_2). \label{eq:26}
\end{align}
In the system \eqref{eq:26}, $X_1$ and $X_2$ represent uninfected and malignant components, respectively, that is, $X_1 = (N, O) \in \mathbb{R}^2$ and $X_2 = (T, V) \in \mathbb{R}^2$. The tumor-free equilibrium is denoted by $E_0$ and define as $E_0 = (X_1^0, 0)$. Thus, the existence of global stability at the tumor-free equilibrium point depends on the following two conditions
\begin{enumerate}
    \item If $\frac{\mathrm{d}X_1}{\mathrm{d}t} = G(X_1, 0)$, $X_1^0$ is globally asymptotically stable.
    \item $H(X_1, X_2) = BX_2 - \bar{H}(X_1, X_2)$, where $\bar{H}(X_1, X_2) \geq 0$ for $(X_1, X_2) \in \Delta$.
\end{enumerate}
At second condition, $B = D_{X_2} H(X_1^0, 0)$ is an $M$-matrix that is the off diagonal entries are positive and $\Delta$ is the feasible region. Then, the following statement holds.\\

\textbf{Lemma 1} \textit{For $\mathcal{R}_0 < 1$, then the equilibrium point $E^0 = (X_1^0, 0)$ of the system (1) is said to be globally asymptotically stable, if the above conditions are satisfied}

\begin{theorem}
\textit{If $\mathcal{R}_0 < 1$, then the proposed model (1) is globally asymptotically stable at disease-free equilibrium $E_0$ and unstable otherwise.}
\end{theorem} 
\begin{proof}
Let $X_1 =N, O)$ and $X_2 = (T, V)$ and define $E_0 = (X_0, 0)$, where
\begin{equation}
    X_1^0 = \left( \frac{\Lambda_n}{\mu_1}, \frac{\Lambda_O}{\mu_2} \right). \label{eq:27}
\end{equation}
By using model system (1), we have
\begin{align*}
    \frac{\mathrm{d}X_1}{\mathrm{d}t} &= G(X_1, X_2), \nonumber \\
    \frac{\mathrm{d}X_1}{\mathrm{d}t} &= \begin{pmatrix}
        \Lambda_n + \gamma_1 VN - \eta_1 TN - \mu_1 N \\
        \Lambda_O + \beta_1 VO - \eta_2 TO - \mu_2 O
    \end{pmatrix}. \label{eq:28}
\end{align*}
For $N = N_0, O = O_0$ and $G(X_1, 0) = 0$, we get
\begin{equation}
    G(X_1, 0) = \begin{pmatrix}
        \Lambda_n - \mu_1 N \\
        \Lambda_O - \mu_2 O
    \end{pmatrix} = 0. \label{eq:29}
\end{equation}
Thus, from equation \eqref{eq:29} as $t \rightarrow \infty$, $X_1 \rightarrow X_1^0$. So $X_1 = X_1^0$ is globally asymptotically stable. Now,
\begin{align*}
    &BX_2 - \bar{H}(X_1, X_2) &= \begin{pmatrix}
        \alpha + \beta O_0 + \gamma N_0 - (\sigma + \mu) & 0 \\
        \delta \left( \frac{K_h}{K_h+O_0} \right) & \epsilon \left( \frac{N_0}{K_n+N_0} \right) - \sigma_1
    \end{pmatrix} \begin{pmatrix} T \\ V \end{pmatrix}\\
    &&- \begin{pmatrix}
        \frac{\alpha T^2}{K} + \beta(O_0 - O)T + \gamma(N_0 - N)T \\
        \delta K_h T \left( \frac{1}{K_h+O_0} - \frac{1}{K_h+O} \right) + \epsilon V \left( \frac{N_0}{K_n+N_0} - \frac{N}{K_n+N} \right)
    \end{pmatrix}.
\end{align*}
As $N_0 \geq N$ and $O_0 \geq O$ so $\frac{\alpha T^2}{K} + \beta(O_0 - O)T + \gamma(N_0 - N)T \geq 0$ and $\delta K_h T \left( \frac{1}{K_h+O_0} - \frac{1}{K_h+O} \right) + \epsilon V \left( \frac{N_0}{K_n+N_0} - \frac{N}{K_n+N} \right) \geq 0$, hence $\bar{H}(X_1, X_2) \geq 0$. Clearly, $B$ is an $M$-matrix and hence both conditions are proved, so by Lemma 1, the disease-free equilibrium point $E_0$ is globally asymptotically stable. 
\end{proof} 
\subsection{Global Asymptotic Stability of the Endemic Equilibrium}
\begin{theorem}
The endemic equilibrium point $E^*(T^*, N^*, O^*, V^*)$ of the system is globally asymptotically stable in the interior of the positive orthant $\mathbb{R}_+^4$ if $R_0 > 1$, provided the system constants satisfy the following constraints.
\begin{align}
    \sigma_1 > \epsilon, \qquad \mu_1 + \eta_1 T^* > \gamma_1 V^*, \qquad \text{and} \qquad \mu_2 + \eta_2 T^* > \beta_1 V^*
    \label{eq:GAS_conditions}
\end{align}
\end{theorem}

\begin{proof}
To analyze the global asymptotic profile, we define the exact Volterra-type nonlinear Lyapunov function $W(T, N, O, V) = \sum_{i=1}^{4} W_i$. 
\begin{align}
    W = \omega_1 \left( T - T^* - T^* \ln \frac{T}{T^*} \right) &+ \omega_2 \left( N - N^* - N^* \ln \frac{N}{N^*} \right) \nonumber \\
    +\, \omega_3 \left( O - O^* - O^* \ln \frac{O}{O^*} \right) &+ \omega_4 \left( V - V^* - V^* \ln \frac{V}{V^*} \right)
\end{align}
where, $\omega_i > 0$ ($i=1,2,3,4$) are positive constant weights. Differentiating $W$ with respect to $t$ we get,
\begin{align}
    \frac{dW}{dt} = \omega_1 \frac{\Delta T}{T} \frac{dT}{dt} + \omega_2 \frac{\Delta N}{N}\frac{dN}{dt} + \omega_3 \frac{\Delta O}{O}\frac{dO}{dt} + \omega_4 \frac{\Delta V}{V}\frac{dV}{dt}
\end{align}
where, $\Delta T = T-T^*$, $\Delta N = N-N^*$, $\Delta O = O-O^*$, and $\Delta V = V-V^*$.

Now, substituting endemic steady-state relations
\begin{align*}
     \sigma +\mu&=\alpha\left(1-\frac{T^*}{K}\right) + \beta O^* + \gamma N^* \\ 
\Lambda_n &= \mu_1 N^* + \eta_1 T^* N^* - \gamma_1 V^* N^* \\
\Lambda_O &= \mu_2 O^* + \eta_2 T^* O^* - \beta_1 V^* O^*
\end{align*}
we get,
\begin{align*}
    \dot W_T &= -\omega_1\frac{\alpha}{K}(\Delta T)^2 +\omega_1\beta\Delta T\Delta O +\omega_1\gamma\Delta T\Delta N \\
    \dot W_N &= -\omega_2 \frac{\mu_1+\eta_1T^*-\gamma_1V^*}{N} (\Delta N)^2 +\omega_2\gamma_1\Delta N\Delta V -\omega_2\eta_1\Delta T\Delta N \\
    \dot W_O &= -\omega_3 \frac{\mu_2+\eta_2T^*-\beta_1V^*}{O} (\Delta O)^2 +\omega_3\beta_1\Delta O\Delta V -\omega_3\eta_2\Delta T\Delta O
\end{align*}

Let, $q(O)=\frac{K_h}{K_h+O}$ and  $h(N)=\frac{N}{K_n+N}$. Then,
\begin{align}
    \dot W_V = \,&\omega_4\delta \frac{\Delta V}{V} \left[ q(O)\Delta T - \frac{K_hT^*\Delta O} {(K_h+O)(K_h+O^*)} \right] \nonumber \\
    &+\omega_4\epsilon \frac{\Delta V}{V} \left[ h(N)\Delta V + \frac{K_nV^*\Delta N} {(K_n+N)(K_n+N^*)} \right] -\omega_4\sigma_1 \frac{(\Delta V)^2}{V}
\end{align}

Finally,
\begin{align}
    \frac{dW}{dt}= \,&-\omega_1\frac{\alpha}{K}(\Delta T)^2 -\omega_2 \frac{\mu_1+\eta_1T^*-\gamma_1V^*}{N} (\Delta N)^2 -\omega_3 \frac{\mu_2+\eta_2T^*-\beta_1V^*}{O} (\Delta O)^2 \nonumber \\
    &-\omega_4 \frac{\sigma_1 - \epsilon h(N)}{V} (\Delta V)^2 + (\omega_1\gamma-\omega_2\eta_1)\Delta T\Delta N + (\omega_1\beta-\omega_3\eta_2)\Delta T\Delta O \nonumber \\
    &+ \omega_2\gamma_1\Delta N\Delta V + \omega_3\beta_1\Delta O\Delta V + \omega_4\delta q(O) \frac{\Delta T\Delta V}{V} \nonumber \\
    &- \omega_4\delta \frac{K_hT^*}{(K_h+O)(K_h+O^*)} \frac{\Delta O\Delta V}{V} + \omega_4\epsilon \frac{K_nV^*}{(K_n+N)(K_n+N^*)} \frac{\Delta N\Delta V}{V}
    \label{eq:exact_Lyapunov}
\end{align}

To resolve the system stability, we choose the constant weights such that $\omega_1\gamma = \omega_2\eta_1$ and $\omega_1\beta = \omega_3\eta_2$, which successfully eliminates the linear $T-N$ and $T-O$ bilinear cross-terms. Since $\sigma_1 > \epsilon$ ensures that the vascular quadratic component remains strictly negative semidefinite.

To bounds the remaining state-dependent nonlinear cross-coupling terms, we apply Young's inequality ($2ab \leq \varepsilon_i a^2 + \frac{1}{\varepsilon_i}b^2$). By selecting bounded positive coefficients $\varepsilon_1, \varepsilon_2,$ and $\varepsilon_3$, these cross-terms are completely absorbed into the negative quadratic coefficients of $(\Delta T)^2$, $(\Delta N)^2$, and $(\Delta O)^2$.\\
Under the parametric bounds established in \eqref{eq:GAS_conditions},
\begin{align}
    \frac{dW}{dt} \leq -D_1 (\Delta T)^2 - D_2 (\Delta N)^2 - D_3 (\Delta O)^2 - D_4 (\Delta V)^2 \leq 0
\end{align}
Now, it is evident that $\frac{dW}{dt} = 0$ holds if and only if $T=T^*$, $N=N^*$, $O=O^*$, and $V=V^*$. Thus, by LaSalle's Invariant Principle, the endemic equilibrium point $E_1$ is globally asymptotically stable under the specified conditions.
\end{proof}

\subsection{Convergence Analysis of Tumor Cells Under Constant Microenvironmental Resources }
To evaluate the asymptotic capability of the tumor cell population under controlled microenvironmental conditions such as continuous oxygenation at a certain level or targeted nutrient perfusion, we utilize a partial decoupling technique. We assume that the microenvironmental nutrient and oxygen levels are held at constant states, denoted by $N = N_c$ and $O = O_c$, respectively. Substituting these constant bounds into the governing equation \eqref{T'} yields
\begin{align}
	\frac {dT}{dt} &= \alpha T\left(1-\frac{T}{K}\right)+\beta O_c T+\gamma N_c T-\sigma T -\mu T \nonumber \\
	&= \left[ (\alpha + \beta O_c + \gamma N_c) - (\sigma + \mu) \right] T - \frac{\alpha}{K}T^2
	\label{eq:T_decouple}
\end{align}
Let us define $\Gamma_c = (\alpha + \beta O_c + \gamma N_c) - (\sigma + \mu)$. Equation \eqref{eq:T_decouple} then reduces to a standard Bernoulli differential equation of logistic representation.
\begin{align}
    \frac {dT}{dt} = \Gamma_c T - \frac{\alpha}{K}T^2 \implies \frac{dT}{T \left( \Gamma_c - \frac{\alpha}{K}T \right)} = dt
    \label{eq:T_separable}
\end{align}
To solve equation \eqref{eq:T_separable}, we resolve the integrand on the left-hand side into partial fraction decomposition as follows:
\begin{align}
    \frac{1}{T \left( \Gamma_c - \frac{\alpha}{K}T \right)} = \frac{1}{\Gamma_c} \left[ \frac{1}{T} + \frac{\frac{\alpha}{K}}{\Gamma_c - \frac{\alpha}{K}T} \right]
\end{align}
Integrating both sides of equation \eqref{eq:T_separable} from the initial state $t=0$ where $T(0) = T_0$ to any arbitrary time $t$, we obtain:
\begin{align}
    \frac{1}{\Gamma_c} \int_{T_0}^{T(t)} \left[ \frac{1}{\xi} + \frac{\frac{\alpha}{K}}{\Gamma_c - \frac{\alpha}{K}\xi} \right] d\xi = \int_{0}^{t} d\tau
\end{align}

\begin{align}
    \implies T(t) = \frac{\frac{K \Gamma_c}{\alpha}}{1 + \left( \frac{K \Gamma_c}{\alpha T_0} - 1 \right) e^{-\Gamma_c t}}
\end{align}
By taking, $t \to \infty $, 
\begin{align}
    \lim_{t \to \infty} T(t) = \frac{K \Gamma_c}{\alpha} = K \left[ 1 + \frac{\beta O_c + \gamma N_c - (\sigma + \mu)}{\alpha} \right]
\end{align}
This result shows a direct connection between tumor growth and microenvironmental resources. Biologically, continuous oxygen and nutrients act as fuel for cancer cells, causing the tumor to expand. However, this growth is completely dependent on these resources. If we apply successful therapies to reduce the level of oxygen and nutrients in the surroundings, the tumor cells will naturally lose their support. Consequently, the tumor cell population will gracefully decline as the resources decline also, while being systematically suppressed by clinical therapies that amplify cellular decay parameters ($\sigma, \mu$).

\subsection{Sensitivity Analysis}
The sensitivity index quantifies the significance of a parameter and the extent to which a parameter influences the model behavior that represents future tumor progression status.
The normalized forward sensitivity index of a variable $P$ with respect to a parameter $\zeta$, is defined as
\begin{align}
   \mathcal{G}_\zeta ^P =\frac{\partial P}{\partial \zeta} \times \frac{\zeta}{P}
    \label{Formula of sensitivity}
\end{align}
In general, the sensitivity indices vary from $-1$ to $1$, Specifically,
\begin{multicols}{2} 
\begin{align*}
   &\mathcal{G}_\alpha ^{\mathcal{R}_0} =\frac{\partial {\mathcal{R}_0}}{\partial \alpha} \times \frac{\alpha}{{\mathcal{R}_0}} =0.95890
   \\
&\mathcal{G}_{\beta} ^{\mathcal{R}_0} =\frac{\partial {\mathcal{R}_0}}{\partial {\beta}} \times \frac{{\beta}}{{\mathcal{R}_0}} =0.01370
\\
&\mathcal{G}_{\gamma} ^{\mathcal{R}_0} =\frac{\partial {\mathcal{R}_0}}{\partial {\gamma}} \times \frac{{\gamma}}{{\mathcal{R}_0}} =0.02740
\\
&\mathcal{G}_{\Lambda_O} ^{\mathcal{R}_0} =\frac{\partial {\mathcal{R}_0}}{\partial {\Lambda_O}} \times \frac{{\Lambda_O}}{{\mathcal{R}_0}} =0.01370
\\
&\mathcal{G}_{\mu_2} ^{\mathcal{R}_0} =\frac{\partial {\mathcal{R}_0}}{\partial {\mu_2}} \times \frac{{\mu_2}}{{\mathcal{R}_0}} =-0.01370
\\
&\mathcal{G}_{\Lambda_n} ^{\mathcal{R}_0} =\frac{\partial {\mathcal{R}_0}}{\partial {\Lambda_n}} \times \frac{{\Lambda_n}}{{\mathcal{R}_0}} =0.02740
\\
&\mathcal{G}_{\mu_1} ^{\mathcal{R}_0} =\frac{\partial {\mathcal{R}_0}}{\partial {\mu_1}} \times \frac{{\mu_1}}{{\mathcal{R}_0}} =-0.02740
\\
&\mathcal{G}_\sigma ^{\mathcal{R}_0} =\frac{\partial {\mathcal{R}_0}}{\partial \sigma} \times \frac{\sigma}{{\mathcal{R}_0}} =-0.16327
\\
&\mathcal{G}_\mu ^{\mathcal{R}_0} =\frac{\partial {\mathcal{R}_0}}{\partial \mu} \times \frac{\mu}{{\mathcal{R}_0}} =-0.83673
\end{align*}
\end{multicols}

\begin{table}[hbt!]
    \centering
    \caption{Sensitivity Indices of $\mathcal{R}_0$ with respect to parameters}
    \setlength{\tabcolsep}{15pt}
    \renewcommand{\arraystretch}{2.0}
    \begin{tabular}{c c c c}
    \hline
      \textbf{Parameter} & \textbf{Index} & \textbf{Parameter} & \textbf{Index}\\
      \hline  
$\alpha$ & $0.95890$ & $\beta$ & $0.01370$ \\ 
$\gamma$& $0.02740$ & $\Lambda_O$& $0.01370$\\
$\mu_2$& $-0.01370$ & $\Lambda_n$ &$0.02740$\\
$\mu_1$ & $-0.02740$ & $\sigma$ &$-0.16327$\\
$\mu$ & $-0.83673$ & & \\
\hline
           \end{tabular}
    \label{Sensitivity_Table}
\end{table}

The sensitivity indices of the model parameters are illustrated graphically in Figure \ref{fig:sensitivity}, based on the values calculated in Table \ref{Sensitivity_Table}. In our proposed model, $\mathcal{R}_0 < 1$ represents the threshold where the tumor cells are successfully eradicated and vanish over time, whereas $\mathcal{R}_0 > 1$ implies that the malignancy will be sustained and progress in the microenvironment. 

The sensitivity indices for the parameters $\alpha, \beta, \gamma, \Lambda_O,$ and $\Lambda_n$ are positive. This demonstrates that reducing the values of these parameters such as lowering the cell division rate $\alpha$ or cutting off the baseline oxygen supply $\Lambda_O$—will pull down the reproduction number and directly lead to tumor extinction. In contrast, the sensitivity indices for the parameters $\mu_2, \mu_1, \sigma,$ and $\mu$ are negative. This suggests that elevating these values such as accelerating the natural death rate of tumor cells $\mu$ or promoting the regression of abnormal blood vessels—exerts a protective force that effectively helps to control tumor growth.

\begin{figure}[hbt!]
	\centering
	\includegraphics[height=8cm,width=10cm]{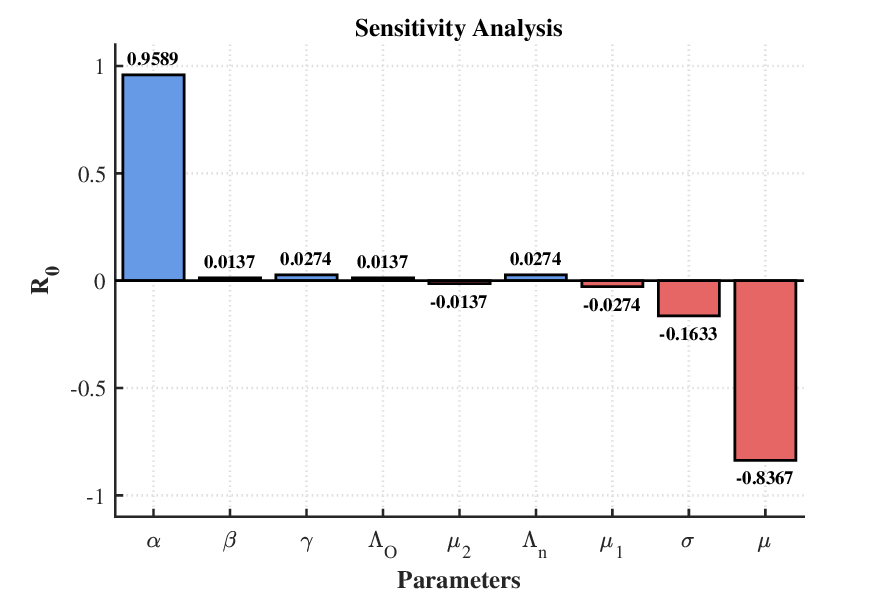}
	\caption{Effects of parameters on $\mathcal{R}_0$}
	\label{fig:sensitivity}
\end{figure}

\section{Numerical Simulation}
To analyze and solve our newly proposed mathematical model numerically, we used MATLAB(R2018a). We solved the system of non linear ordinary differential equations using initial conditions $T(0)=1.2\times10^{9}$ cells, $N(0)=2 molm^{-3}$, $O(0)=40 mmHg$ and $V(0)=15 vessels/mm^{2}$ with the parameter values from Table \ref{table:paramters value} for $5$ years ($60$ months).\\

\subsection{Trajectories and Variation of Parameters}
First, the trajectories or the solutions of our model will be discussed, including the effects of changing various parameters along with their biological meanings. The number of tumor cells $T(t)$ increases steadily until the $20$th month. After reaching close to its carrying capacity $K$, the curve flattens out and starts becoming stable. Biologically, this shows a classic logistic growth pattern. In the beginning, tumor cells grow quickly because they have enough space. But after $20$ months, due to space limitations and biological pressure, the tumor growth reaches its maximum limit and stabilizes at a steady state.
\begin{figure}[hbtp]
	\centering
  \begin{subfigure}{0.49\textwidth}
		\centering
	\includegraphics[width=\textwidth]{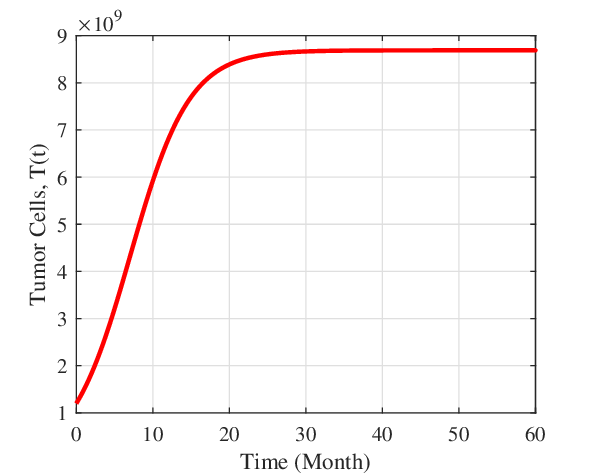}
     \caption{ }
 \label{figure:Tumor_trajectory}
	\end{subfigure} 
 \hfill
	\begin{subfigure}{0.49\textwidth}
\includegraphics[width=\textwidth]{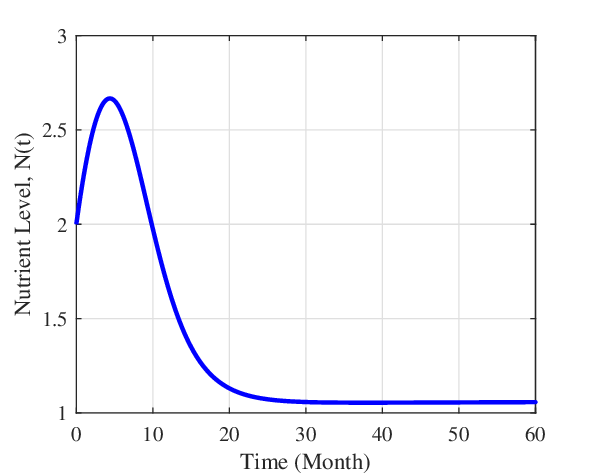}
\caption{ }
	\end{subfigure} 
 \caption{\textbf{(a)} The number of tumor cells $T(t)$ is growing sharply till 20th month. Depending on the natural growth rate $\alpha$, tumor cells reach near the maximum number and stabilize.  \textbf{(b)} Nutrient level N(t) decline to around $1 molm^{-3}$ over the time while tumor cells are growing.}
 \label{Trajectory_T_N}
\end{figure}
\\
In the case of the Nutrient level $N(t)$, the graph shows a small increase during the first $5$ months, reaching up to $2.6molm^{-3}$. This happens because, at the very beginning, the number of tumor cells is small, and they also get nutrition from stromal cells. However, after $5$ months, as the tumor cells grow relentlessly, the nutrient level starts declining steeply and drops close to $1$. This indicates that when tumor cells reach their maximum level and multiply rapidly, the microenvironment suffers from severe nutrient starvation because the tumor cells consume nutrition much faster than the body can supply.\\

\noindent 
\begin{figure}[hbtp]
	\centering
  \begin{subfigure}{0.49\textwidth}
		\centering
	\includegraphics[width=\textwidth]{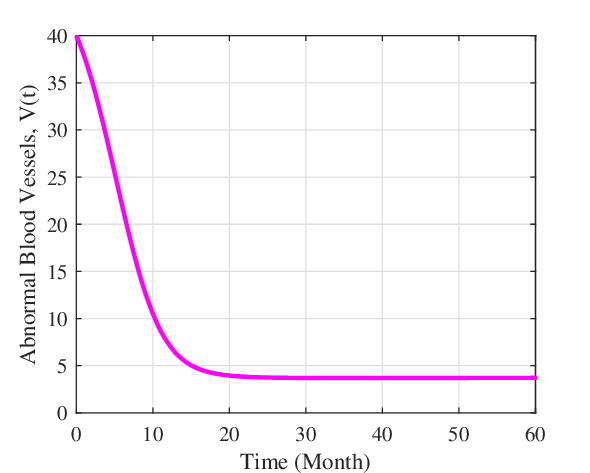}
     \caption{ }
	\end{subfigure} 
 \hfill
	\begin{subfigure}{0.49\textwidth}
\includegraphics[width=\textwidth]{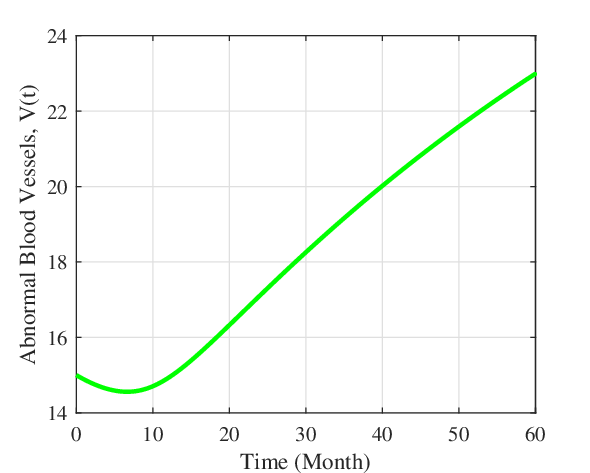}
\caption{ }
	\end{subfigure} 
 \caption{\textbf{(a)} Initially, Oxygen level falls progressively up to just over $10$ months, but after that the declining rate becomes gradual, leading to severe hypoxia in the tumor microenvironment. \textbf{(b)} Abnormal Blood vessels increase along with the reduction of Oxygen level.}
 \label{Trajectory_O_V}
\end{figure}
Similarly, as the tumor cells grow, the Oxygen level $O(t)$ starts declining immediately from $40\text{ mmHg}$. With their increasing size, tumor cells form a dense, solid mass. To meet the rising metabolic demand of these elevated cells, the oxygen level drops rapidly and reaches around $3.8mmHg$ by the $20$th month. This low value indicates a severe hypoxic environment, which is closely associated with the aggressiveness of the tumor and poor clinical outcomes.

Consequently, this severe hypoxia promotes tumor-induced angiogenesis. As the oxygen level drops below a certain level, new blood vessels $V(t)$ are created to save the cancer cells. This is why we see a sharp upward turn in the vessel graph after the $5$th month. Since most of these tumor-induced vessels are abnormal, disorganized, and leaky, they cannot bring uniform blood flow. As a result, an irregular vasculature network is formed in the tumor microenvironment, leading to an uneven and heterogeneous supply of oxygen and nutrients to the cells.
\noindent 
\begin{figure}[hbtp]
	\centering
  \begin{subfigure}{0.49\textwidth}
		\centering
	\includegraphics[width=\textwidth]{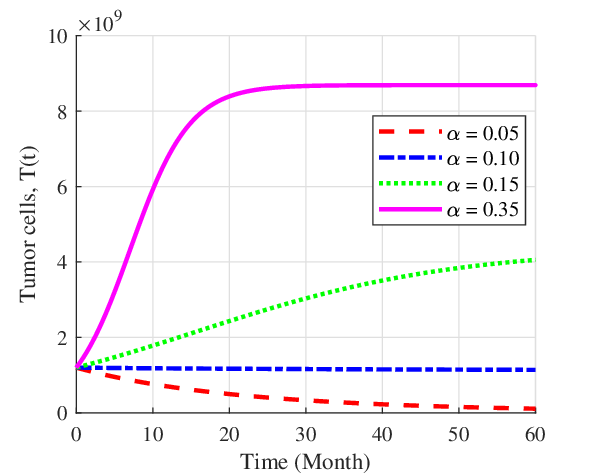}
     \caption{ }
 \label{figure:T_param}
	\end{subfigure} 
 \hfill
 \begin{subfigure}{0.49\textwidth}
		\centering
	\includegraphics[width=\textwidth]{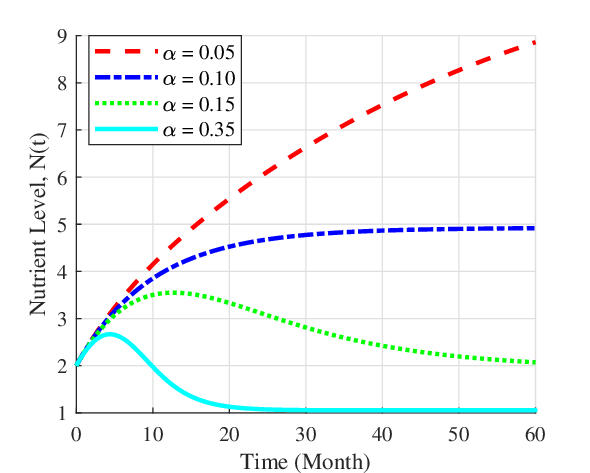}
     \caption{ }
 \label{figure:N_param}
	\end{subfigure} 
    \hfill
    \begin{subfigure}{0.49\textwidth}
		\centering
	\includegraphics[width=\textwidth]{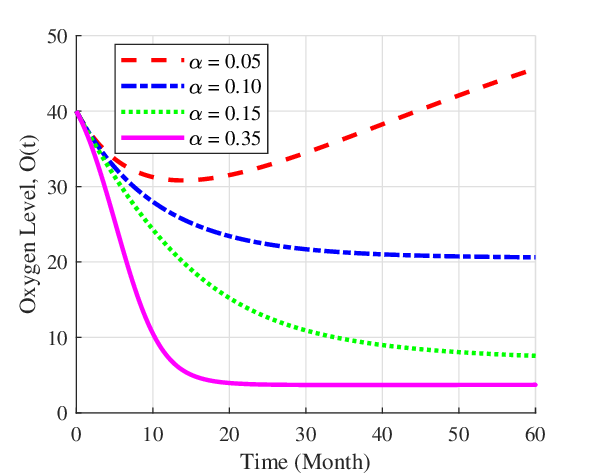}
     \caption{ }
 \label{figure:O_param}
	\end{subfigure}
    \hfill
	\begin{subfigure}{0.49\textwidth}
\includegraphics[width=\textwidth]{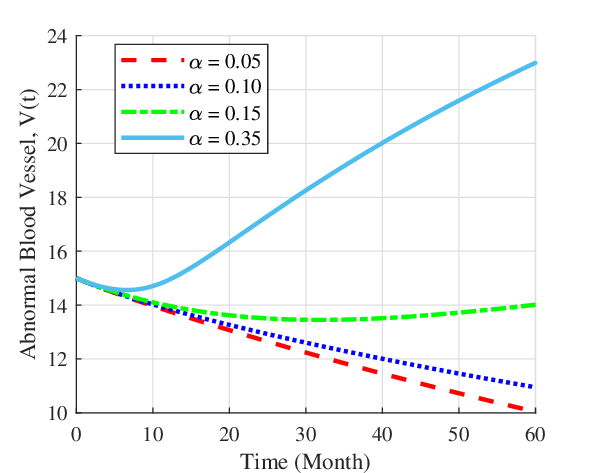}
\caption{ }
	\end{subfigure} 
 \caption{\textbf{(a)} The number of tumor cells $T(t)$ grows when $\alpha>0.10$ and bellow this value it start declining.  \textbf{(b)} Nutrient level N(t) decreases with the higher value of $\alpha$ and when the tumor cells start declining, it starts increasing.\textbf{(c)} A similar pattern to the nutrient level is observed in the case of Oxygen level for different values of $\alpha$. \textbf{(d)} Abnormal Blood vessels increase for the highest value of $\alpha=0.35$ otherwise it declines gradually.}
 \label{Parameter_Variation_alpha}
\end{figure}
\\

Based on different values of the growth parameter $\alpha$, the behavior of the system changes significantly. The natural growth rate $\alpha$ determines both the maximum level of tumor cells and the time required to reach it. When the value of $\alpha$ is high, such as $\alpha = 0.35$, the tumor grows rapidly and reaches a stable steady state after approximately $20$ to $25$ months. On the other hand, when $\alpha = 0.15$, the tumor cells grow moderately and continue to increase past $20$ months without flattening early. However, it is observed that the tumor cells start to decline when the value of $\alpha$ drops below $0.09$. This shows that if the growth rate is below this value, the tumor cells cannot sustain themselves and disappear over time. Biologically, the growth rate can be obtained by the presence of a proliferation protein named $Ki-67$, which is highly active in areas where tumor cells continuously divide. Depending on this growth pattern, Adenoid Cystic Carcinoma (ACC) can be characterized as either a slow-growing or fast-growing tumor.

The resource dynamics in the microenvironment follow a complementary pattern. As the tumor cells grow relentlessly, they consume nutrients and oxygen to sustain their survival. Therefore, a higher value of $\alpha$ leads to a more severe drop in microenvironmental resources. For $\alpha > 0.09$, the nutrient level drops continuously over time. However, when $\alpha$ is below this value and the tumor cells start declining, the nutrient level increases and recovers because there are fewer cells to consume it. An identical pattern is observed in the case of the oxygen level, where a high growth rate ($\alpha = 0.35$) triggers severe hypoxia, dropping the oxygen concentration close to zero, while a low growth rate allows the oxygen level to stay high or recover. 

Finally, the abnormal blood vessels respond dynamically to these changes. For lower growth rates $\alpha \leq 0.10$, where hypoxia is not severe, the number of abnormal vessels gradually decays over the $60$-month period. In contrast, for a high growth rate $\alpha = 0.35$, the severe drop in oxygen levels activates a strong angiogenic response, causing the abnormal blood vessels to continuously grow and enlarge over time.

\begin{figure}[hbtp]
	\centering
  \begin{subfigure}{0.49\textwidth}
		\centering
	\includegraphics[width=\textwidth]{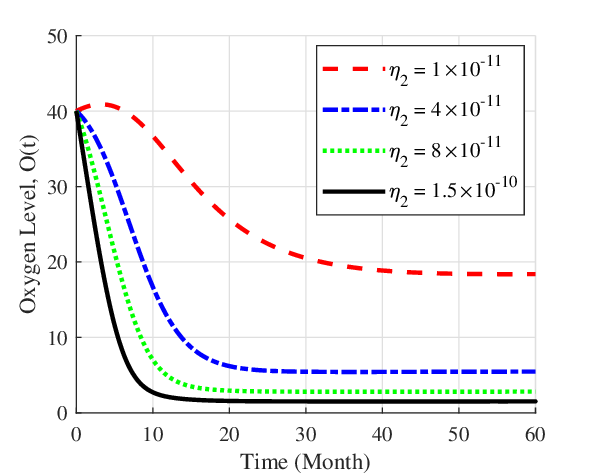}
     \caption{ }
 \label{figure:normoxia}
	\end{subfigure} 
 \hfill
	\begin{subfigure}{0.49\textwidth}
\includegraphics[width=\textwidth]{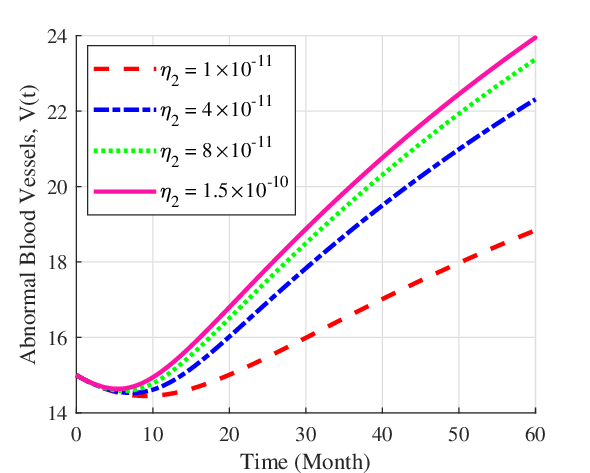}
\caption{ }
	\end{subfigure} 
 \caption{ Oxygen level and blood vessel dynamics vary according to the Variation in the  oxygen consumption rate by tumor cells $\eta_2$ }
 \label{Parameter_variation_gamma1}
\end{figure}

We also investigate the dynamic response of the microenvironment under different values of the oxygen consumption rate by the tumor cells $\eta_2$, as illustrated in Figure 3. The parameter $\eta_2$ determines how aggressively the tumor cells consume oxygen from their surroundings.

As shown in Figure 3(a), when the value of $\eta_2$ is high $\eta_2 = 1.5 \times 10^{-10}$, the tumor cells rapidly deplete the oxygen supply. Consequently, the oxygen level drops steeply within the first $10$ months and stays close to zero. On the other hand, when the consumption rate is low $\eta_2 = 1 \times 10^{-11}$, the drop is much more gradual and the oxygen level safely stabilizes around $18\text{ mmHg}$. This clearly shows that a higher metabolic appetite of tumor cells leads to a more severe and sudden hypoxic stress in the tissue.

This variation in the hypoxic stress directly influences the growth pattern of the vascular network, as depicted in Figure 3(b). For the highest consumption rate $\eta_2 = 1.5 \times 10^{-10}$, the extreme oxygen deprivation activates a powerful angiogenic signal. This causes the abnormal blood vessels $V(t)$ to grow sharply, reaching a maximum count of nearly $24vessels/mm^{2}$ by the $60$th month. In contrast, when the consumption rate is low, the signaling remains weak, restricting the vascular growth to a moderate level. Biologically, this simulation successfully demonstrates that tumors with higher metabolic demands create a more aggressive microenvironment, forcing extensive and rapid abnormal vascularization to sustain their own survival.

\subsection{Phase Portrait Analysis}
To evaluate the long-term qualitative shifts in the system dynamics, we illustrate the phase portraits under different parameter configurations in Figures \ref{Phase_alpha} and \ref{Phase_sigma1}. These trajectories visualize how variation in parameters can drastically alter the tumor microenvironment.

As shown in Figure \ref{Phase_alpha}, the system exhibits a highly sensitive transition depending on the natural tumor growth rate $\alpha$. In Figure \ref{Phase_alpha}(a), the phase portrait between the Nutrient level $N(t)$ and Abnormal Blood Vessels $V(t)$ is depicted, while Figure \ref{Phase_alpha}(b) shows the interaction between the Oxygen level $O(t)$ and Abnormal Blood Vessels $V(t)$. When the growth rate is kept low ($\alpha = 0.05$), the trajectories take a distinct downward path towards a state of tumor eradication, where oxygen levels recover up to $45\text{ mmHg}$ and abnormal blood vessels drop. Conversely, when $\alpha$ is elevated to $0.35$, the curves take a complete opposite turn, looping sharply upward. This mathematical pattern indicates that there is a critical value of cell proliferation rate which aggressively drives the microenvironment into severe hypoxia and extensive vascular progression.\\
A complementary response is observed in Figure \ref{Phase_sigma1} with variation in the vascular regression rate $\sigma_1$. When $\sigma_1 = 0.0015$, the low vascular regression allows tumor-associated blood vessels to persist, resulting in a higher vascular density in the microenvironment. In contrast, increasing $sigma_1$ to $0.15$ increases vascular regression and consequently reduces the persistence of abnormal blood vessels. Thus, the phase trajectories demonstrate that the balance between vascular formation and regression plays an important role in determining the vascular state of the tumor microenvironment.

\begin{figure}[hbtp]
	\centering
  \begin{subfigure}{0.49\textwidth}
		\centering
	\includegraphics[width=\textwidth]{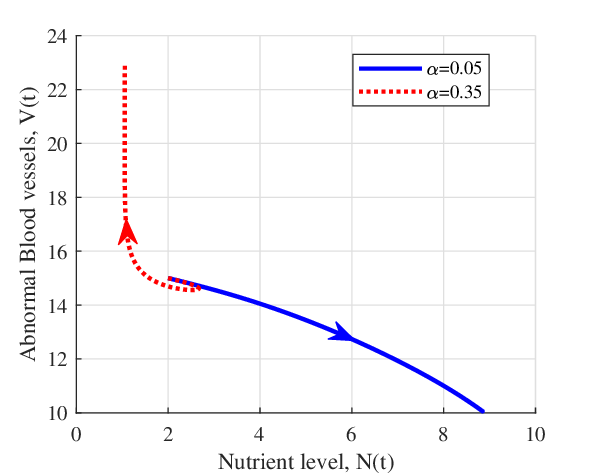}
     \caption{ }
	\end{subfigure} 
 \hfill
	\begin{subfigure}{0.49\textwidth}
\includegraphics[width=\textwidth]{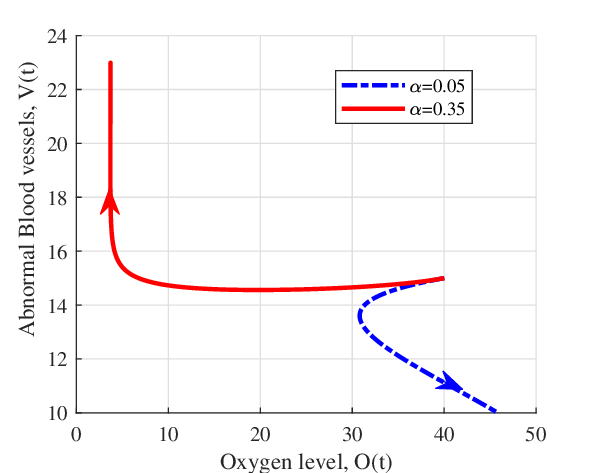}
\caption{ }
	\end{subfigure} 
 \caption{Phase-plane trajectories for two different values of \(\alpha\) in the  \(N\)-\(V\) and \(O\)-\(V\) planes, where the trajectories exhibit different directions, indicating a qualitative change in the behavior of the system}
 \label{Phase_alpha}
\end{figure}

\begin{figure}[hbtp]
	\centering
  \begin{subfigure}{0.49\textwidth}
		\centering
	\includegraphics[width=\textwidth]{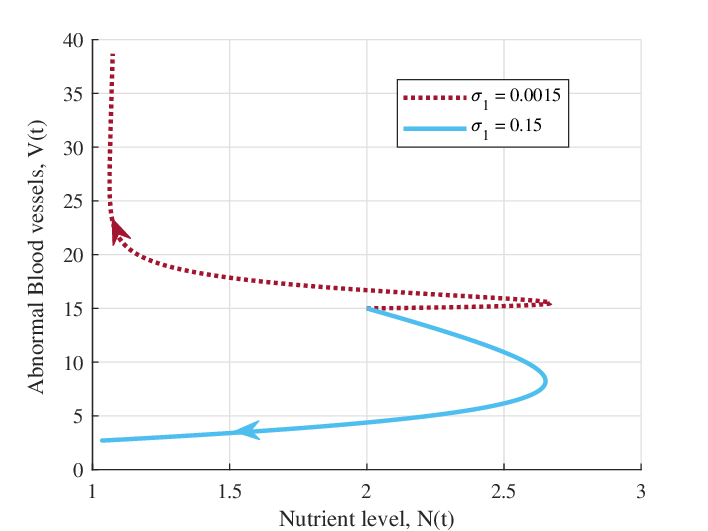}
     \caption{ }
	\end{subfigure} 
 \hfill
	\begin{subfigure}{0.49\textwidth}
\includegraphics[width=\textwidth]{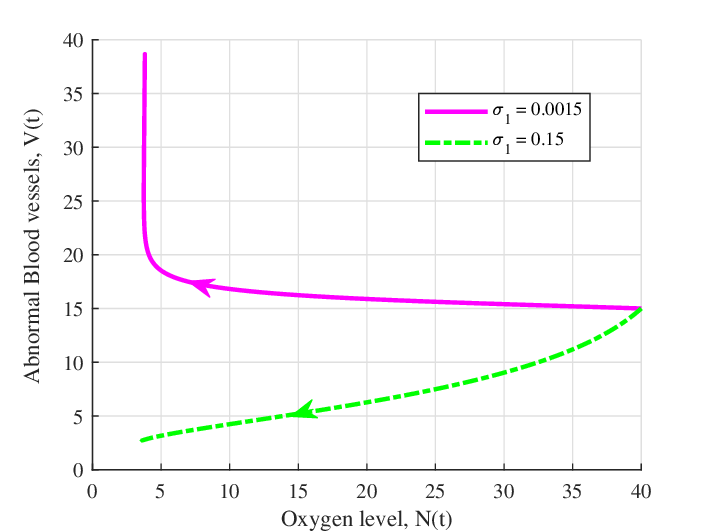}
\caption{ }
	\end{subfigure} 
 \caption{\textbf{(a)} Phase portrait between N(t), V(t) and O(t) and V(t) for the values of the parameter $\sigma_1=0.0015$ and $\sigma_1=0.15$ for which the system poses a two different behavior.}
 \label{Phase_sigma1}
\end{figure}

From a therapeutic perspective, the results obtained from the variation of $\alpha$ and $\sigma_1$ provide a possible interpretation of tumor control in Lacrimal Gland Adenoid Cystic Carcinoma. The phase portraits suggests that maintaining the effective tumor growth rate $\alpha$ below a critical level may prevent the system from remaining in a highly proliferative tumor state. In clinical practice, reduction of tumor burden through surgical intervention followed by appropriate radiotherapy may therefore be interpreted within the model as a strategy for reducing the effective tumor growth and promoting better tumor control.\\ 

Similarly, the value of $\sigma_1$ determines the persistence of tumor associated blood vessels. A low $\sigma_1$ corresponds to limited vascular regression, allowing abnormal vessels to persist and maintain the supply of oxygen and nutrients to the tumor. In contrast, increasing $\sigma_1$ promotes vascular regression and consequently reduces vascular support within the tumor microenvironment. Therefore, the model suggests that targeting tumor-associated vascularization through anti-angiogenic strategies may provide an additional approach for limiting tumor progression. Overall, the findings indicate that controlling both tumor growth and the vascular regression rate may achieve better control of LGACC.
\section{Treatment of Lacrimal Gland Adenoid Cystic Carcinoma (LGACC): Modeling Approach Using Surgery and Radiotherapy}

Lacrimal Gland Adenoid Cystic Carcinoma (LGACC) is a rare and highly aggressive malignant tumor characterized by a high propensity for perineural invasion and distant metastasis. The standard treatment strategy for advanced stages involves surgical resection such as orbital exenteration and eye-preserving surgery followed by adjuvant Postoperative Radiation Therapy. A recent study reported that eye-sparing surgery is increasingly considered a potential treatment strategy for patients with relatively smaller tumors and limited orbital involvement. In contrast, more aggressive features, such as larger tumor size, bone destruction, and perineural tumor spread, were more frequently associated with the need for orbital exenteration \cite{Eye_preserving_26}.

\subsection{Formulation of the Mathematical Model with Interventions}
To construct a realistic clinical treatment timeline, we assume that surgery removes $95\%$ of the bulk tumor cells, leaving a $5\%$ residual tumor burden behind. Although surgical intervention aims to remove the tumor completely, microscopic tumor cells often remain trapped within the orbits of eyes or nerve tracks. Furthermore, clinical guidelines suggest that adjuvant radiotherapy is initiated after approximately 6-8 weeks of surgery to allow for wound healing. Ciric and Sersa \cite{Eva2010} suggested that combining radiotherapy with vascular-disrupting agents can damage abnormal tumor blood vessels and thereby modify the tumor microenvironment, potentially enhancing the sensitivity of tumor cells to radiotherapy. So we consider that the effect of radiotherapy with vascular-disrupting agents can destroy abnormal blood vessels too.

Consequently, the updated mathematical model incorporating treatment is formulated as follows:

\begin{equation}
	\frac {dT}{dt}= \alpha T\left(1-\frac{T}{K}\right)+\beta OT+\gamma NT-\sigma T -\mu T - \xi_{20} T - \psi_{22}T
	\label{T_treat}
\end{equation}

\begin{equation}
	\frac {dN}{dt}= \Lambda_n +\gamma_1 VN -\eta_1 T N-\mu_1N
	\label{N_treat}
\end{equation}

\begin{equation}
	\frac {dO}{dt}=\Lambda_O+\beta_1 VO-\eta_2 T O-\mu_2O
	\label{O_treat}
\end{equation}

\begin{equation}
	\frac {dV}{dt}=\delta\left(\frac{K_h} {K_h+O}\right) T+\epsilon \left(\frac{N}{K_n+N}\right)V-\sigma_1 V - \phi_{22} V
	\label{V_treat}
\end{equation}

Where, the destruction rate of tumor cells due to surgery in the 20th month and radiotherapy after 22th month are represented by $\xi_{20}$ and $\psi_{22}$, respectively. $\phi_{22}$ denotes the disruption rate of abnormal blood vessels due to radiation exposure after 22th month.\\

Basically, Radiotherapy for ACC is commonly delivered as a fractionated treatment regimen, with a total prescribed dose of approximately 60 Gy and a conventional daily fraction of 2 Gy. Accordingly, in our model, radiotherapy is assumed to be applied for 5 consecutive days, followed by a 3-day treatment break after two months of surgery. This treatment schedule is repeated for a total of six cycles. Although, in clinical practice, each radiotherapy session is delivered for a limited duration of approximately 10–15 minutes per day, our model assumes the effect of radiotherapy to be continuous during each treatment period. This simplification is one of the limitations of our model. \\ 

Firstly, surgical intervention is performed at the 20th month.radiotherapy. As the surgery removes $95\%$ of the tumor burden, leaving $5\%$ of the tumor cells as a residual population. In Figure \ref{fig:only_surgery}, it is shown that the residual tumor cells initially remain at a low level after surgery but gradually begin to regrow in the absence of postoperative radiotherapy. This behavior reflects the highly aggressive nature of the tumor, which may exhibit local recurrence and, in some cases, distant metastasis too.

Following surgery, a temporary improvement in the tumor microenvironment can also be observed. With the substantial reduction in tumor burden, the severe hypoxic and nutrient-deprived conditions begin to recover, while the tumor-associated abnormal vasculature also decreases. Consequently, oxygen and nutrient availability can improve for a certain period. However, as the residual tumor cells proliferate again, the adverse microenvironment is gradually re-established, leading to renewed hypoxia, nutrient deprivation, and abnormal vascular development. Thus, the figure demonstrates how surgical removal alone may provide an initial reduction in tumor burden but may not be sufficient to prevent subsequent tumor regrowth.

\begin{figure}[hbt!]
	\centering
	\includegraphics[height=12cm,width=17cm]{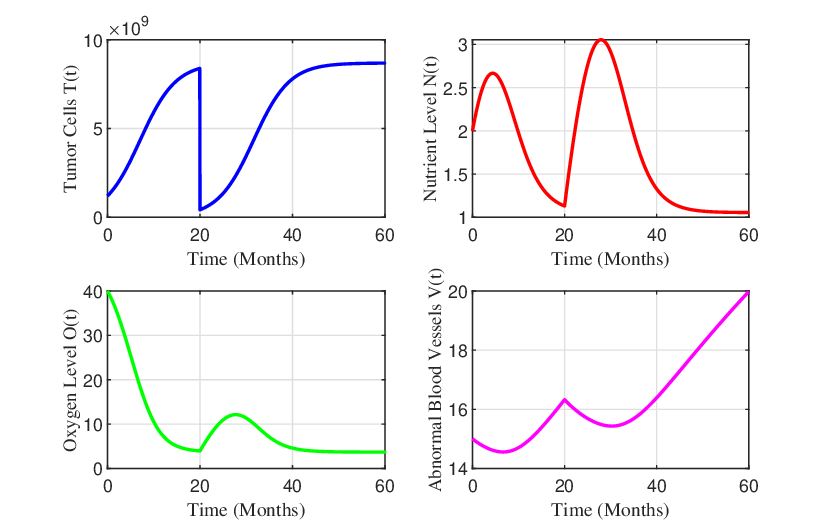}
	\caption{Effect of different radiosensitivity levels on tumor cell dynamics following surgery and radiotherapy.}
	\label{fig:only_surgery}
\end{figure} 

Furthermore, the destruction rate of the remaining tumor cells after surgery due to radiotherapy is not fixed. It depends on the radiosensitivity of the residual tumor cells after surgery. The radiosensitivity may vary depending on the stage at which the tumor is treated and the characteristics of the tumor. For example, in a solid tumor, severe hypoxia within the tumor microenvironment may reduce the effectiveness of radiotherapy. Therefore, in our model, different cases have been illustrated by varying the radiosensitivity of the remaining tumor cells to investigate how changes in radiotherapy sensitivity affect the treatment outcome and overall survival.

To investigate the effect of radiotherapy sensitivity, different radiosensitivity values were considered, as shown in Figure \ref{fig:surgery_with_radiotherapy} to observe the changes in the tumor microenvironment. The results demonstrate that the tumor dynamics vary considerably with changes in radiosensitivity. When the radiosensitivity is relatively low, such as $10\%$ or $20\%$, the residual tumor cells remaining after surgery show substantial regrowth over time, which is characterized as local recurrence of ACC. In contrast, higher radiosensitivity, such as $40\%$, $50\%$, and $70\%$, results in a much greater reduction and better control of the residual tumor burden. In particular, at higher radiosensitivity, the remaining tumor cells are more effectively suppressed following radiotherapy, resulting in limited tumor regrowth. These results indicate that the effectiveness of radiotherapy is strongly influenced by the radiosensitivity of the residual tumor cells.

\begin{figure}[hbt!]
	\centering
\includegraphics[height=15cm,width=18cm]{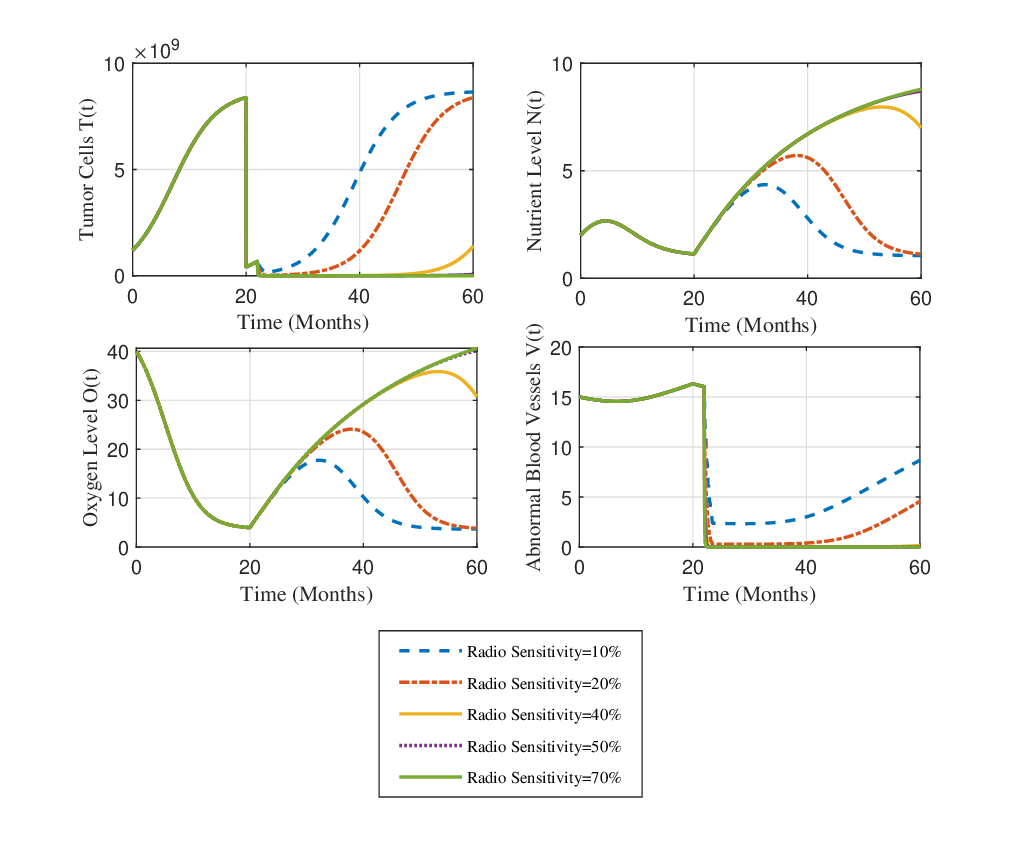}
	\caption{Effect of different radiosensitivity levels on tumor cell dynamics following surgery and radiotherapy.}
	\label{fig:surgery_with_radiotherapy}
\end{figure} 

Additionally, the nutrient and oxygen levels also exhibit different patterns under different radiosensitivity conditions. Following surgery and radiotherapy, these levels initially increase as the tumor burden decreases and the microenvironment begins to recover. However, under lower radiosensitivity, subsequent tumor regrowth is accompanied by a decline in oxygen and nutrient availability, whereas higher radiosensitivity generally maintains better oxygen and nutrient levels over the considered period. Similarly, the abnormal blood vessels decrease sharply following treatment. At lower radiosensitivity, the abnormal vasculature gradually develops again as the residual tumor regrows, while at higher radiosensitivity it remains substantially suppressed. Overall, these results suggest that radiosensitivity influences not only residual tumor growth but also the recovery and subsequent dynamics of the tumor microenvironment.
\begin{center}
  \begin{table}[hbt!]
   \caption{Toxicity of Tumor cells for different values of Radiosensitivity}

    \begin{tabular}{p{3cm} p{4.0cm} p{6.0cm}}
      \hline
      \textbf{Radiosensitivity} & \textbf{Overall Toxicity of Tumor cells} & \textbf{Toxicity of Tumor cells at the end of 60 months}\\
      \hline    
$0\%$ (Only surgery) & $3.5355 \times 10^{11}$ & $8.68283 \times 10^{9}$ \\

$10\%$ & $2.9090 \times 10^{11}$ & $8.6477 \times 10^{9}$ \\

$20\%$ & $2.2255 \times 10^{11}$ & $8.38539 \times 10^{9}$ \\

$40\%$ & $1.1523 \times 10^{11}$ & $1.3801 \times 10^{9}$ \\

$50\%$ & $1.0976 \times 10^{11}$ & $6.64597 \times 10^{7}$ \\

$70\%$ & $1.0948 \times 10^{11}$ & $9271.1$ \\
      \hline
    \end{tabular}
    \label{table:Toxicity}
    \end{table}
  \end{center}
Table \ref{table:Toxicity} shows that tumor cell toxicity differs depending on the various values of radiosensitivity of residual tumor cells remaining after surgery. Higher radiosensitivity results in greater tumor cell destruction and consequently reduces the residual tumor burden. The values presented in the table also demonstrate the effectiveness of postoperative radiotherapy after 60 months, with higher radiosensitivity leading to substantially lower tumor cell populations. In contrast, lower radiosensitivity may result in poor response to radiotherapy and continued tumor progression, which may adversely affect patient survival. The lowest toxicity is observed in the case of $70\%$ radiosensitivity. So, it can be concluded that for better treatment outcomes, surgery should ideally be performed while the residual tumor cells retain relatively high radiosensitivity, before the tumor develops severe hypoxia or hypovascularity. These unfavorable microenvironmental conditions can make the tumor more solid and reduce the radiosensitivity of the remaining tumor cells.
\\
The effect of receiving no postoperative radiotherapy after surgery has already been demonstrated previously. We now consider another possible clinical scenario in which the patient starts radiotherapy but does not maintain the prescribed treatment schedule. For example, the patient may complete the first five days of treatment but subsequently interrupt or discontinue the remaining treatment cycles. Such interruptions may reduce the overall effectiveness of radiotherapy and increase the toxicity of tumor cells. Therefore, we investigate different treatment-interruption scenarios to examine their effects on tumor control and the overall treatment outcome.
\begin{table}[htbp]
\centering
\small
\caption{Overall toxicity and toxicity at the end of 60 months for different radiosensitivity values and skipped doses}
\vspace{2mm}
\begin{tabular}{lcccccc}
\toprule
\diagbox[width=3cm,height=1.2cm]{\raisebox{-2pt}{\scriptsize Skipped time}}{\raisebox{2pt}{\scriptsize Radiosensitivity}} 
& \multicolumn{2}{c}{\textbf{70\%}} 
& \multicolumn{2}{c}{\textbf{40\%}} 
& \multicolumn{2}{c}{\textbf{10\%}} \\
\cmidrule(lr){2-3} \cmidrule(lr){4-5} \cmidrule(lr){6-7}
& \textbf{Overall} & \textbf{Toxicity at} & \textbf{Overall} & \textbf{Toxicity at} & \textbf{Overall} & \textbf{Toxicity at} \\
& \textbf{Toxicity} & \textbf{the end} & \textbf{Toxicity} & \textbf{the end} & \textbf{Toxicity} & \textbf{the end} \\
\midrule
Skip 11 days after dose-1 & $1.0948 \times 10^{11}$ & 313467 & $1.2977 \times 10^{11}$ & $3.9678 \times 10^{9}$ & $3.0127 \times 10^{11}$ & $8.65847 \times 10^{9}$ \\
Skip 19 days after dose-1 & $1.0953 \times 10^{11}$ & $1.05855 \times 10^{7}$ & $1.6115 \times 10^{11}$ & $6.8348 \times 10^{9}$ & $3.1168 \times 10^{11}$ & $8.66646 \times 10^{9}$ \\
\addlinespace
Skip 11 days after dose-2 & $1.0948 \times 10^{11}$ & 313509 & $1.2975 \times 10^{11}$ & $3.97016 \times 10^{9}$ & $3.0128 \times 10^{11}$ & $8.65851 \times 10^{9}$ \\
Skip 19 days after dose-2 & $1.0952 \times 10^{11}$ & $1.05855 \times 10^{7}$ & $1.6114 \times 10^{11}$ & $6.83804 \times 10^{9}$ & $3.1168 \times 10^{11}$ & $8.66651 \times 10^{9}$ \\
\addlinespace
Skip 11 days after dose-3 & $1.0948 \times 10^{11}$ & 313495 & $1.2975 \times 10^{11}$ & $3.9706 \times 10^{9}$ & $3.0128 \times 10^{11}$ & $8.65854 \times 10^{9}$ \\
Skip 19 days after dose-3 & $1.0952 \times 10^{11}$ & $1.05876 \times 10^{7}$ & $1.6144 \times 10^{11}$ & $6.83876 \times 10^{9}$ & $3.1168 \times 10^{11}$ & $8.66656 \times 10^{9}$ \\
\addlinespace
Skip 11 days after dose-4 & $1.0948 \times 10^{11}$ & 313479 & $1.2975 \times 10^{11}$ & $3.9707 \times 10^{9}$ & $3.0128 \times 10^{11}$ & $8.65656 \times 10^{9}$ \\
\bottomrule
\end{tabular}
\label{Toxicity_Table}
\end{table}

Table \ref{Toxicity_Table} illustrates the effect of interrupting the prescribed radiotherapy schedule at different stages of treatment. For a $70\%$ radiosensitivity, when radiotherapy is interrupted for 11 days, the tumor burden at the end of 60 months remains relatively low, with values of $313467$, $313509$, $313495$, and $313479$ when the interruption occurs after Dose--1, Dose--2, Dose--3, and Dose--4, respectively. These values are very close to each other, indicating that an 11-day interruption produces a relatively similar long-term response regardless of which dose is followed by the interruption.

In contrast, when the treatment gap is extended to 19 days, the tumor burden at 60 months increases substantially. For example, after Dose--1, it increases from $3.13\times10^{5}$ for an 11-day interruption to approximately $1.06\times10^{7}$ for a 19-day interruption. Similar increases are observed when the 19-day interruption occurs after Dose--2 or Dose--3. Therefore, the results indicate that a prolonged interruption of radiotherapy has a much greater effect on long-term tumor control than a shorter interruption. This suggests that maintaining the prescribed treatment schedule is important for achieving sustained suppression of the residual tumor cells.

\section{Results and Discussion}
Adenoid Cystic Carcinoma of the lacrimal gland is a highly aggressive tumor. After a certain period of time, management and removal of this tumor become so challenging. One of the most immediate concerns about the tumor is controlling distant metastasis. To make it possible, the tumor has to be exterminated. This study develops a biologically grounded mathematical model incorporating Tumor cells, Oxygen level, Nutrient level and Angiogenesis. The model illustrates key mechanisms of tumor microenvironment, such as oxygen and nutrient-dependent tumor proliferation and hypoxia-induced angiogenesis.\\

At first, in our analytical analyses such as positivity and boundedness, Existence and Uniqueness analyses reveal that our model’s solution exists and is both positive and bounded. In addition, Future status of the tumor represents whether tumor cells will exist or not. Thereafter, we have tested if our model is stable or not both locally and globally and the conditions of stability. Finally, we have performed the convergence test of Tumor cells, keeping Oxygen and nutrients constant. Then observed that the availability of these resources helps the tumor to expand.

Then, we performed numerical simulations and observed that, as the tumor cells increase, the oxygen and nutrient levels gradually decrease due to the increasing demand of the tumor. Because of this resource limitation, the tumor growth may become stabilized after a certain period of time. However, under severe hypoxia and nutrient deprivation, the tumor gradually develops a more solid and unfavorable microenvironment with irregular vasculature structure, making it more difficult to control.

We also investigated the effects of different parameters on the overall system. For example, the natural growth rate of tumor cells, $\alpha$, and the oxygen consumption rate of tumor cells, $\eta_2$, were found to have noticeable effects on the tumor dynamics. Moreover, we identified some important parameters that can cause qualitative changes in the whole system. In particular, $\alpha$ and $\sigma_1$, where $\sigma_1$ represents the natural regression rate of abnormal blood vessels, have critical values that can determine whether the tumor persists or eventually disappears.

Finally, because controlling this tumor is challenging, we investigated the effects of eye-sparing surgery followed by postoperative radiotherapy. Our results show that tumor recurrence after surgery depends strongly on the radiosensitivity of the residual tumor cells and on whether the prescribed radiotherapy schedule is properly maintained. Among the considered radiosensitivity levels, higher radiosensitivity, particularly 70\%, resulted in much better tumor control, with a very low tumor burden observed over the considered period. We also investigated treatment interruption period and found that increasing the treatment gap from 11 days to 19 days caused a substantial increase in the tumor burden at the end of 60 months.

Overall, these results suggest that surgery alone may not be sufficient to prevent tumor regrowth when residual tumor cells remain. Postoperative radiotherapy can play an important role in controlling these residual cells, but its effectiveness depends on the radiosensitivity of the tumor and maintaining the prescribed treatment schedule.\\

As a part of future research, optimal control strategies can be introduced into the proposed model to investigate optimal therapeutic strategies.\\

\begin{center}
  \begin{table}[hbt!]
   \caption{Comparison of related studies and the present work, highlighting its novelty}

    \begin{tabular}{p{1.8cm} p{2.2cm} p{3.0cm} p{6.0cm}}
      \hline
      \textbf{Study} & \textbf{Methodology} &\textbf{Focus Area} & \textbf{Novelty} \\
      \hline
Dillon et al. \cite{Dillon2015} 
& Molecular genetic analysis 
& Genetic and chromosomal alterations of ACC 
& Identified the chromosomal translocation (6:9) leading to MYB--NFIB fusion and highlighted MYB dysregulation as an important molecular feature of ACC. 
\\

Lin et al. \cite{Lin2015} 
& Molecular overview 
& Clinical and biological behavior of ACC 
& Described ACC as a slow-growing but highly infiltrative tumor with a tendency for intracranial invasion and poor prognosis. Developed the LACC-1 cell line to investigate the growth of ACC in vitro. 
\\

Kozal and Krześlak \cite{Kozal2022} 
& Molecular review 
& Effect of hypoxia in tumor microenvironment 
& Reported that hypoxia can promote immune suppression, metabolic adaptation, angiogenesis, cell proliferation, and metastasis, highlighting its important role in tumor progression. 
\\

Ghaffari et al. \cite{Ghaffari2015}
& Molecular overview 
& Tumor metabolism and nutrients 
& Demonstrated that tumor cells can reprogram metabolic pathways to support proliferation, survival, and metastasis, with these processes being strongly influenced by nutrient availability. 
\\

Chiche et al. \cite{Rouleau2010}
& Clinical analysis 
& Hypoxia-induced chemoresistance 
& Showed that hypoxia-induced mitochondrial remodeling is associated with chemotherapy resistance and that restoring mitochondrial function can increase drug sensitivity. 
\\

Present study 
& Nonlinear ODE-based dynamical system 
& Microenvironmental dynamics and treatment response in LGACC 
& Developed a mathematical model to investigate the interactions among tumor growth, oxygen, nutrients and abnormal blood vessels in lacrimal gland ACC. The model further examines how microenvironmental conditions influence tumor progression and the response of residual tumor cells to postoperative radiotherapy. 
\\
\hline
    \end{tabular}
    \label{table:Novelty of current work}
    \end{table}
  \end{center}
  
\section*{Data Availability Statement}
The data used to support the findings of this study are included in the article.
\section*{Conflicts of Interest}
The authors declare that they have no known financial or interpersonal conflicts that would have caused it to have an impact on the study described in this article.
\section*{Funding}
The research did not receive any specific funding.






\bibliography{sn-bibliography}

\end{document}